\documentclass{llncs}
\usepackage{amssymb,latexsym,pslatex}
\usepackage{graphics}
\usepackage{pspicture}
\usepackage{epsfig}
\usepackage{multirow}
\newtheorem{fact}[theorem]{Fact}

\newcommand{\braced}[1]{{\left\{#1\right\}}}

\newcommand{\fst}{\textit{fst}}
\newcommand{\lst}{\textit{lst}}

\title{Faster 3-Periodic Merging Networks}

\author{Marek Piotr\'ow}
\institute{Institute of Computer Science, University of Wroc\l aw,\\ 
ul.~Joliot-Curie~15, PL-50-383 Wroc\l aw, Poland\\
\email{Marek.Piotrow@ii.uni.wroc.pl}}

\begin{document}

\maketitle

\begin{abstract} 
We consider the problem of merging two sorted sequences on a comparator
network that is used repeatedly, that is, if the output is not sorted,
the network is applied again using the output as input. The challenging
task is to construct such networks of small depth. The first
constructions of merging networks with a constant period were given by
Kuty{\l}owski, Lory{\'s} and Oesterdikhoff \cite{klo}. They have given
$3$-periodic network that merges two sorted sequences of $N$ numbers in
time $12\log N$ and a similar network of period $4$ that works in
$5.67\log N$.  We present a new family of such networks that are based
on Canfield and Williamson periodic sorter \cite{cw}. Our $3$-periodic
merging networks work in time upper-bounded by $6\log N$. The construction
can be easily generalized to larger constant periods with decreasing
running time, for example, to $4$-periodic ones that work in time 
upper-bounded by $4\log N$. Moreover, to obtain the facts we have introduced 
a new proof technique.
\end{abstract}




\section{Introduction}  
\label{intro}

Comparator networks are probably the simplest parallel model that is used to
solve such tasks as sorting, merging or selecting \cite{k}. Each network
represents a data-oblivious algorithm, which can be easily implemented in
hardware. Moreover, sorting networks can be applied in secure, multi-party
computation (SMC) protocols. They are also strongly connected with switching
networks \cite{l}. The most famous constructions of sorting networks are
Odd-Even and Bitonic networks of depth $\frac{1}{2}\log^2 N$ due to Batcher
\cite{b} and AKS networks of depth $O(\log N)$ due to Ajtai, Komlos and
Szemeredi \cite{aks}. The long-standing disability to decrease a large
constant hidden behind the asymptotically optimal complexity of AKS networks
to a practical value \cite{s} has resulted in studying easier, sorting-related
problems, whose optimal networks have small constants.

A comparator network consists of a set of $N$ registers, each of which can
contain an item from a totally ordered set, and a sequence of
comparator stages.  Each stage is a set of comparators that connect
disjoint pairs of registers and, therefore, can work in parallel (a
comparator is a simple device that takes a contents of two registers
and performs a compare-exchange operation on them: the minimum is put
into the first register and the maximum into the second one). Stages
are run one after another in synchronous manner, hence we can consider
the number of stages as the running time. The size of a network is 
defined to be the total number of comparators in all its stages.

A network $A$ consisting of stages $S_1,S_2,\ldots,S_d$ is called $p$-periodic
if $p<d$ and for each $i$, $1\le i\le d-p$, stages $S_i$ and $S_{i+p}$ are
identical.  A periodic network is easy to implement, especially in hardware,
because one can use the first $p$ stages in a cycle: if the output of $p$-th
stage is not correct (sorted, for example), the sequence of $p$ stages is run
again. We can also define a $p$-periodic network just by giving the total
number of stages and a description of its first $p$ stages. A challenging task
is to construct a family of small-periodic networks for sorting-related
problems with the running time equal to, or not much greater than that of
non-periodic networks.

Dowd et al.\ \cite{dpsr} gave the construction of $\log N$-periodic sorting
networks of $N$ registers with running time of $\log^2 N$.  Kuty{\l}owski et
al.\ \cite{klow} introduced a general method to convert a non-periodic sorting
network into a 5-periodic one, but the running time increases by a factor of
$O(\log N)$ during the conversion. For simpler problems such as merging or
correction there are constant-periodic networks that solve the corresponding
problem in asymptotically optimal logarithmic time \cite{klo,p}. In
particular, Kuty{\l}owski, Lory{\'s} and Oesterdikhoff \cite{klo} have given
$3$-periodic network that merges two sorted sequences of $N$ numbers in time
$12\log N$ and a similar network of period $4$ that works in $5.67\log N$.
They have also sketched a construction of merging networks with periods larger
than 4 and running time decreasing asymptotically to $2.25\log N$.

In this paper, we introduce a new family of constant-periodic merging
networks that are based on the Canfield and Williamson $O(\log
N)$-periodic sorter \cite{cw} by a certain periodification technique.
Our $3$-periodic merging networks work in time upper-bounded by $6\log
N$ and $4$-periodic ones - in time upper-bounded by $4\log N$. The
construction can be easily generalized to larger constant periods with
decreasing running time.

The advantage of constant-periodic networks is that they have pretty
simple patterns of communication links, that is, each node (register)
of such a network can only be connected to a constant number of other
nodes. Such patterns are easier to implement, for example, in hardware.
Moreover, a node uses these links in a simple periodic manner and this
can save control login and simplify timing considerations.

\section{Periodic merging networks}

Our merging networks are based on the Canfield and Williamson \cite{cw}
$O(\log N)$-periodic sorters. We recall now the definition of their
networks: for each $k\ge 1$ let $CW_k=S_1,\ldots,S_k$ denote a network of
$N=2^k$ registers, where the stages are defined as follows (see also
Figures \ref{mergeCW} and \ref{anotherCW}):
\begin{eqnarray}
S_1 & = & \braced{[2i:2i+1]: i=0,1,\ldots,2^{k-1}-1}, \\
S_{j+1} & = & \braced{[2i+1:2i+2^{k-j}]:
  i=0,1,\ldots,2^{k-1}-2^{k-j-1}-1}, j=1,\ldots,k-1.
\end{eqnarray}
\begin{figure}[ht]
\begin{center}
\begin{picture}(164,99)
\newcounter{reg}
\thicklines
\setcounter{reg}{28}
\multiput(0,3)(0,12){8}{
	\put(-4,7){\makebox(0,0)[br]{\footnotesize\arabic{reg}}} 
	\put(0,0){\line(1,0){164}} 
	\put(0,3){\line(1,0){164}} 
	\put(0,6){\line(1,0){164}} 
	\put(0,9){\line(1,0){164}} 
	\addtocounter{reg}{-4}}
\put(-4,2){\makebox(0,0)[br]{\footnotesize 31}}
\setlength{\arrowlength}{2pt}
\begin{picture}(164,99)
\put(20,96){\vector(0,-1){3}}
\put(20,90){\vector(0,-1){3}}
\put(20,84){\vector(0,-1){3}}
\put(20,78){\vector(0,-1){3}}
\put(20,72){\vector(0,-1){3}}
\put(20,66){\vector(0,-1){3}}
\put(20,60){\vector(0,-1){3}}
\put(20,54){\vector(0,-1){3}}
\put(20,48){\vector(0,-1){3}}
\put(20,42){\vector(0,-1){3}}
\put(20,36){\vector(0,-1){3}}
\put(20,30){\vector(0,-1){3}}
\put(20,24){\vector(0,-1){3}}
\put(20,18){\vector(0,-1){3}}
\put(20,12){\vector(0,-1){3}}
\put(20,6){\vector(0,-1){3}}
\put(40,93){\vector(0,-1){45}}
\put(44,87){\vector(0,-1){45}}
\put(48,81){\vector(0,-1){45}}
\put(52,75){\vector(0,-1){45}}
\put(56,69){\vector(0,-1){45}}
\put(60,63){\vector(0,-1){45}}
\put(64,57){\vector(0,-1){45}}
\put(68,51){\vector(0,-1){45}}
\put(88,93){\vector(0,-1){21}}
\put(92,87){\vector(0,-1){21}}
\put(96,81){\vector(0,-1){21}}
\put(100,75){\vector(0,-1){21}}
\put(88,69){\vector(0,-1){21}}
\put(92,63){\vector(0,-1){21}}
\put(96,57){\vector(0,-1){21}}
\put(100,51){\vector(0,-1){21}}
\put(88,45){\vector(0,-1){21}}
\put(92,39){\vector(0,-1){21}}
\put(96,33){\vector(0,-1){21}}
\put(100,27){\vector(0,-1){21}}
\put(120,93){\vector(0,-1){9}}
\put(124,87){\vector(0,-1){9}}
\put(120,81){\vector(0,-1){9}}
\put(124,75){\vector(0,-1){9}}
\put(120,69){\vector(0,-1){9}}
\put(124,63){\vector(0,-1){9}}
\put(120,57){\vector(0,-1){9}}
\put(124,51){\vector(0,-1){9}}
\put(120,45){\vector(0,-1){9}}
\put(124,39){\vector(0,-1){9}}
\put(120,33){\vector(0,-1){9}}
\put(124,27){\vector(0,-1){9}}
\put(120,21){\vector(0,-1){9}}
\put(124,15){\vector(0,-1){9}}
\put(144,93){\vector(0,-1){3}}
\put(144,87){\vector(0,-1){3}}
\put(144,81){\vector(0,-1){3}}
\put(144,75){\vector(0,-1){3}}
\put(144,69){\vector(0,-1){3}}
\put(144,63){\vector(0,-1){3}}
\put(144,57){\vector(0,-1){3}}
\put(144,51){\vector(0,-1){3}}
\put(144,45){\vector(0,-1){3}}
\put(144,39){\vector(0,-1){3}}
\put(144,33){\vector(0,-1){3}}
\put(144,27){\vector(0,-1){3}}
\put(144,21){\vector(0,-1){3}}
\put(144,15){\vector(0,-1){3}}
\put(144,9){\vector(0,-1){3}}
\thinlines
\put(30,0){\line(0,1){99}}
\put(78,0){\line(0,1){99}}
\put(110,0){\line(0,1){99}}
\put(134,0){\line(0,1){99}}
\end{picture}
\end{picture}
\end{center}
\caption{The Canfield and Williamson $\log N$-periodic sorter, where
  $N=32$. Registers and comparators are represented by horizontal lines and 
  arrows, respectively. Stages are separated by vertical lines.}
\label{mergeCW}
\end{figure}
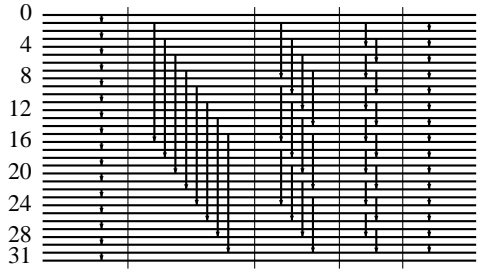

The merging and sorting properties of the networks are given in the
following proposition.

\begin{proposition}
(1) For each $k\ge 1$, if two sorted sequences of length
$2^{k-1}$ are given in registers with odd and even indices, respectively, 
then $CW_k$ is a merging network.
(2) For each $k\ge 1$, $CW_k$ is a $k$-pass periodic sorting network. 
\end{proposition}

\begin{figure}[ht]
\begin{center}
\begin{picture}(210,160)
\includegraphics{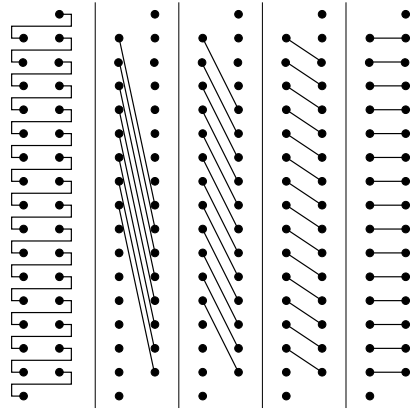}
\end{picture}
\end{center}
\caption{Another view of $CW_5$ 5-pass 5-periodic sorter. Registers and
  comparators are represented by dots and edges, respectively. Stages
  are separated by vertical lines.}
\label{anotherCW}
\end{figure}
We would like to implement a version of this network as a
constant-periodic comparator network. Consider first the most
challenging 3-periodic implementation. We start with the definition of a
temporally construction $P_k$ which structure is similar to the
structure of $CW_k$. Then we transform it to 3-periodic network $M_k$.
The idea is to replace each register $i$ in $CW_k$ (except the first and
the last ones) with a sequence of $k-2$ consecutive registers, move the
endpoints of long comparators one register further or closer depending
on the parity of $i$ and insert between each pair of stages containing
long comparators a stage with short comparators joining the endpoints of
those long ones. The result is depicted in Fig. \ref {merge92}. In this
way, we obtain a network in which each register is used in at most
three consecutive stages. Therefore the network $P_k$ can be packed into
the first 3 stages and used periodically to get the desired 3-periodic
merging network.
\begin{figure}[ht]
\begin{center}
\begin{picture}(320,100)
\includegraphics[scale=0.8]{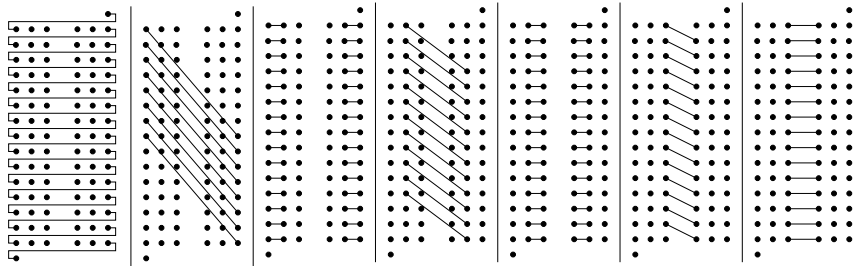}
\end{picture}
\end{center}
\caption{$P_5$ as an implementation of $CW_5$. Registers and comparators
  are represented by dots and edges, respectively. Stages are separated
  by vertical lines. Stages with short horizontal comparators are
  inserted between stages with long comparators.}
\label{merge92}
\end{figure}

Let $[i:j]$ denote a comparator connecting registers $i$ and $j$. A
comparator $[i:j]$ is {\em standard} if $i<j$.  For an $N$-register
network $A = S_1,S_2,\ldots,S_d$, where $S_1,S_2,\ldots,S_{d}$ denote
stages, and for an integer $j\in\{1,\ldots,N\}$, we will use the
following notations:
\begin{eqnarray}
   \fst(j,A) &  = & \min\braced{1\le i\le d: j \in regs(S_i)}\\
   \lst(j,A) &  = & \max\braced{1\le i\le d: j \in regs(S_i)}\\
   delay(A) &  = & \max_{j\in\{1,\ldots,N\}}\braced{\lst(j,A)-\fst(j,A)+1}
\end{eqnarray}
where $regs(\{[i_1:j_1], \ldots, [i_r:j_r]\})$ denotes the set
$\{i_1, j_1, \ldots, i_r, j_r\}$.

Let us define formally the new family of merging networks.  For each
$k\ge 3$ we would like to transform the network $CW_k$ into a new
network $P_k$. 
\begin{definition}\label{defMk}
  Let $n_k=2^{k-1}-1$ be one less than the half of the number of
  registers in $CW_k$ and $b_k=2(k-2)$. The number of registers of $P_k$
  is defined to be $N_k=n_k\cdot b_k+2$. The stages of $P_k =
  S_{k,1}\cup\{[0:1],[N_k-2:N_k-1]\}, S_{k,2},\ldots,S_{k,2k-3}$ are defined by
  the following equations, where $j=1,\ldots,\frac{b_k}{2}$:
\begin{eqnarray}
S_{k,1} & = & \braced{[b_ki:b_ki+1]: i=1,\ldots,n_k-1}\\
S_{k,2j} & = & \braced{[b_ki+j:b_k(i+2^{k-j-1}-1)+(b_k-j+1)]: 
          i=0,\ldots,n_k-2^{k-j-1}}\\ 
S_{k,2j+1} & = & \left\{[b_ki+j:b_ki+j+1],\right.\\
         &   & \left. [b_ki+(b_k-j):b_ki+(b_k-j+1)]: i=0,\ldots,n_k-1 \right\}
\end{eqnarray}
\end{definition}

The network $P_5$ is depicted in Figure \ref{merge3}.

\begin{figure}[ht]
\begin{center}
\setcounter{reg}{88}
\begin{picture}(204,279) 
\multiput(0,3)(0,12){23}{
	\put(-4,7){\makebox(0,0)[br]{\footnotesize\arabic{reg}}} 
	\put(0,0){\line(1,0){204}} 
	\put(0,3){\line(1,0){204}} 
	\put(0,6){\line(1,0){204}} 
	\put(0,9){\line(1,0){204}} 
	\addtocounter{reg}{-4}}
\put(-4,2){\makebox(0,0)[br]{\footnotesize 91}}
\setlength{\arrowlength}{2pt}
\begin{picture}(204,279)
\setlength{\arrowlength}{2pt}
\put(20,276){\vector(0,-1){3}}
\put(20,258){\vector(0,-1){3}}
\put(20,240){\vector(0,-1){3}}
\put(20,222){\vector(0,-1){3}}
\put(20,204){\vector(0,-1){3}}
\put(20,186){\vector(0,-1){3}}
\put(20,168){\vector(0,-1){3}}
\put(20,150){\vector(0,-1){3}}
\put(20,132){\vector(0,-1){3}}
\put(20,114){\vector(0,-1){3}}
\put(20,96){\vector(0,-1){3}}
\put(20,78){\vector(0,-1){3}}
\put(20,60){\vector(0,-1){3}}
\put(20,42){\vector(0,-1){3}}
\put(20,24){\vector(0,-1){3}}
\put(20,6){\vector(0,-1){3}}
\put(40,273){\vector(0,-1){141}}
\put(44,255){\vector(0,-1){141}}
\put(48,237){\vector(0,-1){141}}
\put(52,219){\vector(0,-1){141}}
\put(56,201){\vector(0,-1){141}}
\put(60,183){\vector(0,-1){141}}
\put(64,165){\vector(0,-1){141}}
\put(68,147){\vector(0,-1){141}}
\put(88,273){\vector(0,-1){3}}
\put(88,261){\vector(0,-1){3}}
\put(88,255){\vector(0,-1){3}}
\put(88,243){\vector(0,-1){3}}
\put(88,237){\vector(0,-1){3}}
\put(88,225){\vector(0,-1){3}}
\put(88,219){\vector(0,-1){3}}
\put(88,207){\vector(0,-1){3}}
\put(88,201){\vector(0,-1){3}}
\put(88,189){\vector(0,-1){3}}
\put(88,183){\vector(0,-1){3}}
\put(88,171){\vector(0,-1){3}}
\put(88,165){\vector(0,-1){3}}
\put(88,153){\vector(0,-1){3}}
\put(88,147){\vector(0,-1){3}}
\put(88,135){\vector(0,-1){3}}
\put(88,129){\vector(0,-1){3}}
\put(88,117){\vector(0,-1){3}}
\put(88,111){\vector(0,-1){3}}
\put(88,99){\vector(0,-1){3}}
\put(88,93){\vector(0,-1){3}}
\put(88,81){\vector(0,-1){3}}
\put(88,75){\vector(0,-1){3}}
\put(88,63){\vector(0,-1){3}}
\put(88,57){\vector(0,-1){3}}
\put(88,45){\vector(0,-1){3}}
\put(88,39){\vector(0,-1){3}}
\put(88,27){\vector(0,-1){3}}
\put(88,21){\vector(0,-1){3}}
\put(88,9){\vector(0,-1){3}}
\put(108,270){\vector(0,-1){63}}
\put(112,252){\vector(0,-1){63}}
\put(116,234){\vector(0,-1){63}}
\put(120,216){\vector(0,-1){63}}
\put(108,198){\vector(0,-1){63}}
\put(112,180){\vector(0,-1){63}}
\put(116,162){\vector(0,-1){63}}
\put(120,144){\vector(0,-1){63}}
\put(108,126){\vector(0,-1){63}}
\put(112,108){\vector(0,-1){63}}
\put(116,90){\vector(0,-1){63}}
\put(120,72){\vector(0,-1){63}}
\put(140,270){\vector(0,-1){3}}
\put(140,264){\vector(0,-1){3}}
\put(140,252){\vector(0,-1){3}}
\put(140,246){\vector(0,-1){3}}
\put(140,234){\vector(0,-1){3}}
\put(140,228){\vector(0,-1){3}}
\put(140,216){\vector(0,-1){3}}
\put(140,210){\vector(0,-1){3}}
\put(140,198){\vector(0,-1){3}}
\put(140,192){\vector(0,-1){3}}
\put(140,180){\vector(0,-1){3}}
\put(140,174){\vector(0,-1){3}}
\put(140,162){\vector(0,-1){3}}
\put(140,156){\vector(0,-1){3}}
\put(140,144){\vector(0,-1){3}}
\put(140,138){\vector(0,-1){3}}
\put(140,126){\vector(0,-1){3}}
\put(140,120){\vector(0,-1){3}}
\put(140,108){\vector(0,-1){3}}
\put(140,102){\vector(0,-1){3}}
\put(140,90){\vector(0,-1){3}}
\put(140,84){\vector(0,-1){3}}
\put(140,72){\vector(0,-1){3}}
\put(140,66){\vector(0,-1){3}}
\put(140,54){\vector(0,-1){3}}
\put(140,48){\vector(0,-1){3}}
\put(140,36){\vector(0,-1){3}}
\put(140,30){\vector(0,-1){3}}
\put(140,18){\vector(0,-1){3}}
\put(140,12){\vector(0,-1){3}}
\put(160,267){\vector(0,-1){21}}
\put(164,249){\vector(0,-1){21}}
\put(160,231){\vector(0,-1){21}}
\put(164,213){\vector(0,-1){21}}
\put(160,195){\vector(0,-1){21}}
\put(164,177){\vector(0,-1){21}}
\put(160,159){\vector(0,-1){21}}
\put(164,141){\vector(0,-1){21}}
\put(160,123){\vector(0,-1){21}}
\put(164,105){\vector(0,-1){21}}
\put(160,87){\vector(0,-1){21}}
\put(164,69){\vector(0,-1){21}}
\put(160,51){\vector(0,-1){21}}
\put(164,33){\vector(0,-1){21}}
\put(184,267){\vector(0,-1){3}}
\put(184,249){\vector(0,-1){3}}
\put(184,231){\vector(0,-1){3}}
\put(184,213){\vector(0,-1){3}}
\put(184,195){\vector(0,-1){3}}
\put(184,177){\vector(0,-1){3}}
\put(184,159){\vector(0,-1){3}}
\put(184,141){\vector(0,-1){3}}
\put(184,123){\vector(0,-1){3}}
\put(184,105){\vector(0,-1){3}}
\put(184,87){\vector(0,-1){3}}
\put(184,69){\vector(0,-1){3}}
\put(184,51){\vector(0,-1){3}}
\put(184,33){\vector(0,-1){3}}
\put(184,15){\vector(0,-1){3}}
\end{picture}
\end{picture}
\end{center}
\caption{The traditional drawing of $P_5$ network}
\label{merge3}
\end{figure}
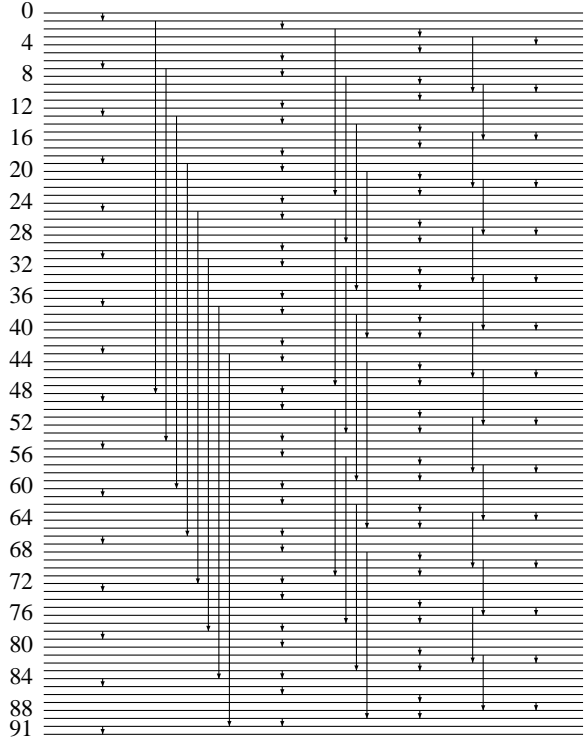

\begin{fact}
$delay(P_k) = 3$ for $ k \ge 3$. \qed
\end{fact}

Let $A = S_1,S_2,\ldots,S_d$ and $A' = S'_1,S'_2,\ldots,S'_{d'}$ be
$N$-input comparator networks such that for each $i$, $1\le i \le\min(d,d')$,
$regs(S_i) \cap regs(S'_i) = \emptyset$.  Then
$A\cup A'$ is defined to be
$(S_1\cup S'_1),(S_2\cup S'_2),\ldots,(S_{\max(d,d')}\cup S'_{\max(d,d')})$,
where empty stages are added at the end of the network of smaller depth.

For any comparator network $A=S_1,\ldots,S_d$ and
$D = delay(A)$, let us define a network $B=T_1,\ldots,T_D$ to be a {\em 
compact form} of $A$, where $T_q=\bigcup\braced{S_{q+pD}: 0\le p\le(d-q)/D}$, 
$1\le q\le D$. Observe that $B$ is correctly defined due to the delay
of $A$.  Moreover, $depth(B)=delay(B)=delay(A)$.

\begin{definition}
For $k\ge 3$ let $M_k$ denote the compact form of $P_k$ with the first
and the last registers deleted. That is, the network $M_k=T^k_1,T^k_2,T^k_3$ 
is using the set of registers numbered $\{1,2, \ldots, N_k\}$, where 
$N_k=(2^{k-1}-1)\cdot 2(k-2)$, and $T^k_j=\{S_{k,j+3i}: 0\le 
i\le\frac{2k-j-3}{3}\}$, $j=1,2,3$.
\end{definition}
It is not necessary to delete the first and the last registers of $P_k$ but 
this will simplify proofs a little bit in the next section. The network $M_5$ 
is given in Fig. \ref{merge3p}.
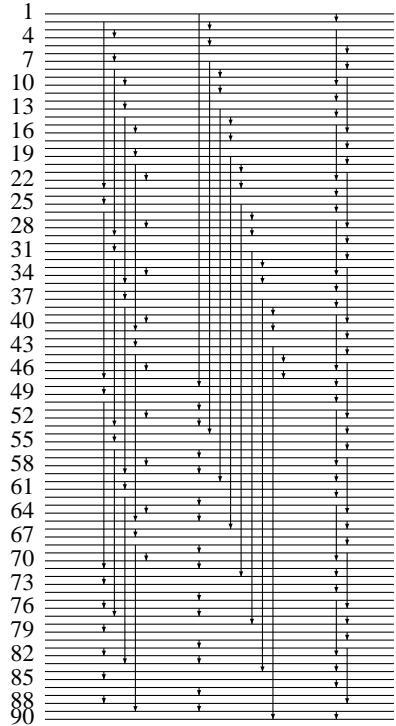
\begin{figure}[ht]
\begin{center}
\setcounter{reg}{88}
\begin{picture}(132,279)
\multiput(0,3)(0,18){15}{
	\put(-4,4){\makebox(0,0)[br]{\footnotesize\arabic{reg}}} 
	\addtocounter{reg}{-3}
	\put(-4,13){\makebox(0,0)[br]{\footnotesize\arabic{reg}}} 
	\put(0,0){\line(1,0){132}} 
	\put(0,3){\line(1,0){132}} 
	\put(0,6){\line(1,0){132}} 
	\put(0,9){\line(1,0){132}} 
	\put(0,12){\line(1,0){132}} 
	\put(0,15){\line(1,0){132}} 
	\addtocounter{reg}{-3}}
\put(-4,1){\makebox(0,0)[br]{\footnotesize 90}}
\setlength{\arrowlength}{2pt}
\begin{picture}(132,279)(0,3)
\setlength{\arrowlength}{2pt}
\put(20,270){\vector(0,-1){63}}
\put(24,267){\vector(0,-1){3}}
\put(24,258){\vector(0,-1){3}}
\put(24,252){\vector(0,-1){63}}
\put(28,249){\vector(0,-1){3}}
\put(28,240){\vector(0,-1){3}}
\put(28,234){\vector(0,-1){63}}
\put(32,231){\vector(0,-1){3}}
\put(32,222){\vector(0,-1){3}}
\put(32,216){\vector(0,-1){63}}
\put(36,213){\vector(0,-1){3}}
\put(20,204){\vector(0,-1){3}}
\put(20,198){\vector(0,-1){63}}
\put(36,195){\vector(0,-1){3}}
\put(24,186){\vector(0,-1){3}}
\put(24,180){\vector(0,-1){63}}
\put(36,177){\vector(0,-1){3}}
\put(28,168){\vector(0,-1){3}}
\put(28,162){\vector(0,-1){63}}
\put(36,159){\vector(0,-1){3}}
\put(32,150){\vector(0,-1){3}}
\put(32,144){\vector(0,-1){63}}
\put(36,141){\vector(0,-1){3}}
\put(20,132){\vector(0,-1){3}}
\put(20,126){\vector(0,-1){63}}
\put(36,123){\vector(0,-1){3}}
\put(24,114){\vector(0,-1){3}}
\put(24,108){\vector(0,-1){63}}
\put(36,105){\vector(0,-1){3}}
\put(28,96){\vector(0,-1){3}}
\put(28,90){\vector(0,-1){63}}
\put(36,87){\vector(0,-1){3}}
\put(32,78){\vector(0,-1){3}}
\put(32,72){\vector(0,-1){63}}
\put(36,69){\vector(0,-1){3}}
\put(20,60){\vector(0,-1){3}}
\put(20,51){\vector(0,-1){3}}
\put(20,42){\vector(0,-1){3}}
\put(20,33){\vector(0,-1){3}}
\put(20,24){\vector(0,-1){3}}
\put(20,15){\vector(0,-1){3}}
\put(56,273){\vector(0,-1){141}}
\put(60,270){\vector(0,-1){3}}
\put(60,264){\vector(0,-1){3}}
\put(60,255){\vector(0,-1){141}}
\put(64,252){\vector(0,-1){3}}
\put(64,246){\vector(0,-1){3}}
\put(64,237){\vector(0,-1){141}}
\put(68,234){\vector(0,-1){3}}
\put(68,228){\vector(0,-1){3}}
\put(68,219){\vector(0,-1){141}}
\put(72,216){\vector(0,-1){3}}
\put(72,210){\vector(0,-1){3}}
\put(72,201){\vector(0,-1){141}}
\put(76,198){\vector(0,-1){3}}
\put(76,192){\vector(0,-1){3}}
\put(76,183){\vector(0,-1){141}}
\put(80,180){\vector(0,-1){3}}
\put(80,174){\vector(0,-1){3}}
\put(80,165){\vector(0,-1){141}}
\put(84,162){\vector(0,-1){3}}
\put(84,156){\vector(0,-1){3}}
\put(84,147){\vector(0,-1){141}}
\put(88,144){\vector(0,-1){3}}
\put(88,138){\vector(0,-1){3}}
\put(56,126){\vector(0,-1){3}}
\put(56,120){\vector(0,-1){3}}
\put(56,108){\vector(0,-1){3}}
\put(56,102){\vector(0,-1){3}}
\put(56,90){\vector(0,-1){3}}
\put(56,84){\vector(0,-1){3}}
\put(56,72){\vector(0,-1){3}}
\put(56,66){\vector(0,-1){3}}
\put(56,54){\vector(0,-1){3}}
\put(56,48){\vector(0,-1){3}}
\put(56,36){\vector(0,-1){3}}
\put(56,30){\vector(0,-1){3}}
\put(56,18){\vector(0,-1){3}}
\put(56,12){\vector(0,-1){3}}
\put(108,273){\vector(0,-1){3}}
\put(108,267){\vector(0,-1){21}}
\put(112,261){\vector(0,-1){3}}
\put(112,255){\vector(0,-1){3}}
\put(112,249){\vector(0,-1){21}}
\put(108,243){\vector(0,-1){3}}
\put(108,237){\vector(0,-1){3}}
\put(108,231){\vector(0,-1){21}}
\put(112,225){\vector(0,-1){3}}
\put(112,219){\vector(0,-1){3}}
\put(112,213){\vector(0,-1){21}}
\put(108,207){\vector(0,-1){3}}
\put(108,201){\vector(0,-1){3}}
\put(108,195){\vector(0,-1){21}}
\put(112,189){\vector(0,-1){3}}
\put(112,183){\vector(0,-1){3}}
\put(112,177){\vector(0,-1){21}}
\put(108,171){\vector(0,-1){3}}
\put(108,165){\vector(0,-1){3}}
\put(108,159){\vector(0,-1){21}}
\put(112,153){\vector(0,-1){3}}
\put(112,147){\vector(0,-1){3}}
\put(112,141){\vector(0,-1){21}}
\put(108,135){\vector(0,-1){3}}
\put(108,129){\vector(0,-1){3}}
\put(108,123){\vector(0,-1){21}}
\put(112,117){\vector(0,-1){3}}
\put(112,111){\vector(0,-1){3}}
\put(112,105){\vector(0,-1){21}}
\put(108,99){\vector(0,-1){3}}
\put(108,93){\vector(0,-1){3}}
\put(108,87){\vector(0,-1){21}}
\put(112,81){\vector(0,-1){3}}
\put(112,75){\vector(0,-1){3}}
\put(112,69){\vector(0,-1){21}}
\put(108,63){\vector(0,-1){3}}
\put(108,57){\vector(0,-1){3}}
\put(108,51){\vector(0,-1){21}}
\put(112,45){\vector(0,-1){3}}
\put(112,39){\vector(0,-1){3}}
\put(112,33){\vector(0,-1){21}}
\put(108,27){\vector(0,-1){3}}
\put(108,21){\vector(0,-1){3}}
\put(108,9){\vector(0,-1){3}}
\end{picture}
\end{picture}
\end{center}
\caption{The $M_5$ network}
\label{merge3p}
\end{figure}
\begin{theorem} \label{3merger} There exists a family of 3-periodic
  comparator networks $M_k$, $k\ge 3$, such that each $M_k$ is a
  $2k-5$-pass merger of two sorted sequences given in odd and even
  registers, respectively. The running time of $M_k$ is $6k-15\le 6\log
  N_k$, where $N_k=(2^k-2)(k-2)$ is the number of registers in $M_k$.
\end{theorem}

The proof is based on the observation that $M_k$ merges $k-2$ pairs of
sorted subsequences, one after another, in pipeline fashion. Details
are given in the next section. 

In a similar way, we can convert $CW_k$ into a 4-periodic merging
network. Assume that $k$ is even. We replace each register (except the
first and the last ones) with a sequence of $(k-2)/2$ consecutive
registers, move the endpoints of long comparators in such a way that
exactly two long comparators start or end at each new register and
insert after each pair of stages containing long comparators a stage
with short comparators joining the endpoints of those long
comparators. The result is depicted in Fig. \ref{merge4}.

\begin{figure}[ht]
\begin{center}
\setcounter{reg}{120}
\begin{picture}(264,381) 
\multiput(0,9)(0,12){31}{
	\put(-4,7){\makebox(0,0)[br]{\footnotesize\arabic{reg}}} 
	\put(0,0){\line(1,0){264}} 
	\put(0,3){\line(1,0){264}} 
	\put(0,6){\line(1,0){264}} 
	\put(0,9){\line(1,0){264}} 
	\addtocounter{reg}{-4}}
\put(0,3){\line(1,0){264}} 
\put(0,6){\line(1,0){264}} 
\put(-4,2){\makebox(0,0)[br]{\footnotesize 125}}
\setlength{\arrowlength}{2pt}
\begin{picture}(264,381)
\put(20,378){\vector(0,-1){3}}
\put(20,366){\vector(0,-1){3}}
\put(20,354){\vector(0,-1){3}}
\put(20,342){\vector(0,-1){3}}
\put(20,330){\vector(0,-1){3}}
\put(20,318){\vector(0,-1){3}}
\put(20,306){\vector(0,-1){3}}
\put(20,294){\vector(0,-1){3}}
\put(20,282){\vector(0,-1){3}}
\put(20,270){\vector(0,-1){3}}
\put(20,258){\vector(0,-1){3}}
\put(20,246){\vector(0,-1){3}}
\put(20,234){\vector(0,-1){3}}
\put(20,222){\vector(0,-1){3}}
\put(20,210){\vector(0,-1){3}}
\put(20,198){\vector(0,-1){3}}
\put(20,186){\vector(0,-1){3}}
\put(20,174){\vector(0,-1){3}}
\put(20,162){\vector(0,-1){3}}
\put(20,150){\vector(0,-1){3}}
\put(20,138){\vector(0,-1){3}}
\put(20,126){\vector(0,-1){3}}
\put(20,114){\vector(0,-1){3}}
\put(20,102){\vector(0,-1){3}}
\put(20,90){\vector(0,-1){3}}
\put(20,78){\vector(0,-1){3}}
\put(20,66){\vector(0,-1){3}}
\put(20,54){\vector(0,-1){3}}
\put(20,42){\vector(0,-1){3}}
\put(20,30){\vector(0,-1){3}}
\put(20,18){\vector(0,-1){3}}
\put(20,6){\vector(0,-1){3}}
\put(40,375){\vector(0,-1){189}}
\put(44,363){\vector(0,-1){189}}
\put(48,351){\vector(0,-1){189}}
\put(52,339){\vector(0,-1){189}}
\put(56,327){\vector(0,-1){189}}
\put(60,315){\vector(0,-1){189}}
\put(64,303){\vector(0,-1){189}}
\put(68,291){\vector(0,-1){189}}
\put(72,279){\vector(0,-1){189}}
\put(76,267){\vector(0,-1){189}}
\put(80,255){\vector(0,-1){189}}
\put(84,243){\vector(0,-1){189}}
\put(88,231){\vector(0,-1){189}}
\put(92,219){\vector(0,-1){189}}
\put(96,207){\vector(0,-1){189}}
\put(100,195){\vector(0,-1){189}}
\put(120,375){\vector(0,-1){93}}
\put(124,363){\vector(0,-1){93}}
\put(128,351){\vector(0,-1){93}}
\put(132,339){\vector(0,-1){93}}
\put(136,327){\vector(0,-1){93}}
\put(140,315){\vector(0,-1){93}}
\put(144,303){\vector(0,-1){93}}
\put(148,291){\vector(0,-1){93}}
\put(120,279){\vector(0,-1){93}}
\put(124,267){\vector(0,-1){93}}
\put(128,255){\vector(0,-1){93}}
\put(132,243){\vector(0,-1){93}}
\put(136,231){\vector(0,-1){93}}
\put(140,219){\vector(0,-1){93}}
\put(144,207){\vector(0,-1){93}}
\put(148,195){\vector(0,-1){93}}
\put(120,183){\vector(0,-1){93}}
\put(124,171){\vector(0,-1){93}}
\put(128,159){\vector(0,-1){93}}
\put(132,147){\vector(0,-1){93}}
\put(136,135){\vector(0,-1){93}}
\put(140,123){\vector(0,-1){93}}
\put(144,111){\vector(0,-1){93}}
\put(148,99){\vector(0,-1){93}}
\put(168,375){\vector(0,-1){3}}
\put(168,369){\vector(0,-1){3}}
\put(168,363){\vector(0,-1){3}}
\put(168,357){\vector(0,-1){3}}
\put(168,351){\vector(0,-1){3}}
\put(168,345){\vector(0,-1){3}}
\put(168,339){\vector(0,-1){3}}
\put(168,333){\vector(0,-1){3}}
\put(168,327){\vector(0,-1){3}}
\put(168,321){\vector(0,-1){3}}
\put(168,315){\vector(0,-1){3}}
\put(168,309){\vector(0,-1){3}}
\put(168,303){\vector(0,-1){3}}
\put(168,297){\vector(0,-1){3}}
\put(168,291){\vector(0,-1){3}}
\put(168,285){\vector(0,-1){3}}
\put(168,279){\vector(0,-1){3}}
\put(168,273){\vector(0,-1){3}}
\put(168,267){\vector(0,-1){3}}
\put(168,261){\vector(0,-1){3}}
\put(168,255){\vector(0,-1){3}}
\put(168,249){\vector(0,-1){3}}
\put(168,243){\vector(0,-1){3}}
\put(168,237){\vector(0,-1){3}}
\put(168,231){\vector(0,-1){3}}
\put(168,225){\vector(0,-1){3}}
\put(168,219){\vector(0,-1){3}}
\put(168,213){\vector(0,-1){3}}
\put(168,207){\vector(0,-1){3}}
\put(168,201){\vector(0,-1){3}}
\put(168,195){\vector(0,-1){3}}
\put(168,189){\vector(0,-1){3}}
\put(168,183){\vector(0,-1){3}}
\put(168,177){\vector(0,-1){3}}
\put(168,171){\vector(0,-1){3}}
\put(168,165){\vector(0,-1){3}}
\put(168,159){\vector(0,-1){3}}
\put(168,153){\vector(0,-1){3}}
\put(168,147){\vector(0,-1){3}}
\put(168,141){\vector(0,-1){3}}
\put(168,135){\vector(0,-1){3}}
\put(168,129){\vector(0,-1){3}}
\put(168,123){\vector(0,-1){3}}
\put(168,117){\vector(0,-1){3}}
\put(168,111){\vector(0,-1){3}}
\put(168,105){\vector(0,-1){3}}
\put(168,99){\vector(0,-1){3}}
\put(168,93){\vector(0,-1){3}}
\put(168,87){\vector(0,-1){3}}
\put(168,81){\vector(0,-1){3}}
\put(168,75){\vector(0,-1){3}}
\put(168,69){\vector(0,-1){3}}
\put(168,63){\vector(0,-1){3}}
\put(168,57){\vector(0,-1){3}}
\put(168,51){\vector(0,-1){3}}
\put(168,45){\vector(0,-1){3}}
\put(168,39){\vector(0,-1){3}}
\put(168,33){\vector(0,-1){3}}
\put(168,27){\vector(0,-1){3}}
\put(168,21){\vector(0,-1){3}}
\put(168,15){\vector(0,-1){3}}
\put(168,9){\vector(0,-1){3}}
\put(188,372){\vector(0,-1){39}}
\put(192,360){\vector(0,-1){39}}
\put(196,348){\vector(0,-1){39}}
\put(200,336){\vector(0,-1){39}}
\put(188,324){\vector(0,-1){39}}
\put(192,312){\vector(0,-1){39}}
\put(196,300){\vector(0,-1){39}}
\put(200,288){\vector(0,-1){39}}
\put(188,276){\vector(0,-1){39}}
\put(192,264){\vector(0,-1){39}}
\put(196,252){\vector(0,-1){39}}
\put(200,240){\vector(0,-1){39}}
\put(188,228){\vector(0,-1){39}}
\put(192,216){\vector(0,-1){39}}
\put(196,204){\vector(0,-1){39}}
\put(200,192){\vector(0,-1){39}}
\put(188,180){\vector(0,-1){39}}
\put(192,168){\vector(0,-1){39}}
\put(196,156){\vector(0,-1){39}}
\put(200,144){\vector(0,-1){39}}
\put(188,132){\vector(0,-1){39}}
\put(192,120){\vector(0,-1){39}}
\put(196,108){\vector(0,-1){39}}
\put(200,96){\vector(0,-1){39}}
\put(188,84){\vector(0,-1){39}}
\put(192,72){\vector(0,-1){39}}
\put(196,60){\vector(0,-1){39}}
\put(200,48){\vector(0,-1){39}}
\put(220,372){\vector(0,-1){15}}
\put(224,360){\vector(0,-1){15}}
\put(220,348){\vector(0,-1){15}}
\put(224,336){\vector(0,-1){15}}
\put(220,324){\vector(0,-1){15}}
\put(224,312){\vector(0,-1){15}}
\put(220,300){\vector(0,-1){15}}
\put(224,288){\vector(0,-1){15}}
\put(220,276){\vector(0,-1){15}}
\put(224,264){\vector(0,-1){15}}
\put(220,252){\vector(0,-1){15}}
\put(224,240){\vector(0,-1){15}}
\put(220,228){\vector(0,-1){15}}
\put(224,216){\vector(0,-1){15}}
\put(220,204){\vector(0,-1){15}}
\put(224,192){\vector(0,-1){15}}
\put(220,180){\vector(0,-1){15}}
\put(224,168){\vector(0,-1){15}}
\put(220,156){\vector(0,-1){15}}
\put(224,144){\vector(0,-1){15}}
\put(220,132){\vector(0,-1){15}}
\put(224,120){\vector(0,-1){15}}
\put(220,108){\vector(0,-1){15}}
\put(224,96){\vector(0,-1){15}}
\put(220,84){\vector(0,-1){15}}
\put(224,72){\vector(0,-1){15}}
\put(220,60){\vector(0,-1){15}}
\put(224,48){\vector(0,-1){15}}
\put(220,36){\vector(0,-1){15}}
\put(224,24){\vector(0,-1){15}}
\put(244,372){\vector(0,-1){3}}
\put(244,360){\vector(0,-1){3}}
\put(244,348){\vector(0,-1){3}}
\put(244,336){\vector(0,-1){3}}
\put(244,324){\vector(0,-1){3}}
\put(244,312){\vector(0,-1){3}}
\put(244,300){\vector(0,-1){3}}
\put(244,288){\vector(0,-1){3}}
\put(244,276){\vector(0,-1){3}}
\put(244,264){\vector(0,-1){3}}
\put(244,252){\vector(0,-1){3}}
\put(244,240){\vector(0,-1){3}}
\put(244,228){\vector(0,-1){3}}
\put(244,216){\vector(0,-1){3}}
\put(244,204){\vector(0,-1){3}}
\put(244,192){\vector(0,-1){3}}
\put(244,180){\vector(0,-1){3}}
\put(244,168){\vector(0,-1){3}}
\put(244,156){\vector(0,-1){3}}
\put(244,144){\vector(0,-1){3}}
\put(244,132){\vector(0,-1){3}}
\put(244,120){\vector(0,-1){3}}
\put(244,108){\vector(0,-1){3}}
\put(244,96){\vector(0,-1){3}}
\put(244,84){\vector(0,-1){3}}
\put(244,72){\vector(0,-1){3}}
\put(244,60){\vector(0,-1){3}}
\put(244,48){\vector(0,-1){3}}
\put(244,36){\vector(0,-1){3}}
\put(244,24){\vector(0,-1){3}}
\put(244,12){\vector(0,-1){3}}
\end{picture}
\end{picture}
\end{center}
\caption{The $P'_6$ network}
\label{merge4}
\end{figure}
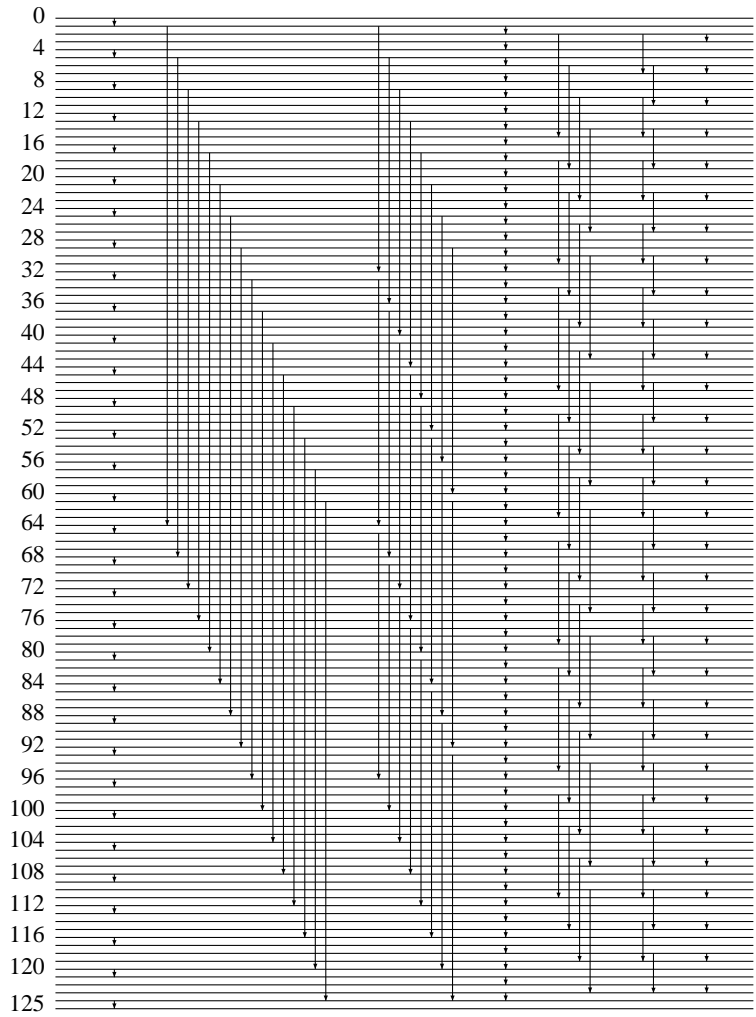

\section{Proof of Theorem \ref{3merger}}  

The first observation we would like to make is that we can consider
inputs consisting of 0's and 1's only. The well-known Zero-One Principle
states that any comparator network that sorts 0-1 input sequences
correctly sorts also arbitrary input sequences \cite{k}.  In the similar
way, we can prove that the same property holds also for merging:
\begin{proposition}
If a comparator network merges any two 0-1 sorted sequences, then it
correctly merges any two sorted sequences. \qed
\end{proposition} 

Therefore we can analyze computations of the network $M_k$, $k\ge 3$, by
describing each state of registers as a 0-1 sequence $\overline{x}=(x_1,
\ldots, x_{N_k})$, where $x_i$ represents the content of register
$i$. If $\overline{x}$ is an input sequence for $2k-5$ passes of $M_k$,
then by $\overline{x}^{(i)}$ we denote the content of registers after $i$
passes of $M_k$, $i=0,\ldots,2k-5$,, that is,
$\overline{x}^{(0)} = \overline{x}$ and $\overline{x}^{(i+1)} =
M_k(\overline{x}^{(i)})$. Since $M_k$ consists of three stages $T^k_1$,
$T^k_2$ and $T^k_3$, we extend the notation to describe the output of each
stage: $\overline{x}^{(i,0)} = \overline{x}^{(i)}$ and
$\overline{x}^{(i,j)} = T^k_j(\overline{x}^{(i,j-1)})$, for
$j=1,2,3$. For other values of $j$ we assume that $\overline{x}^{(i,j)}
= \overline{x}^{(i+j~div~3,j~mod~3)}$. We will use this superscript
notation for other equivalent representations of sequence
$\overline{x}$. 

Now let us fix some technical notations and definitions. A 0-1
sequence can be represented as a word over $\Sigma=\{0,1\}$. A
non-decreasing (also called {\em sorted}) 0-1 sequence has a form of
$0^*1^*$ and can be equivalently represented by the number of ones (or
zeros) in it. For any $x\in\Sigma^*$ let $ones(x)$ denote the number
of $1$ in $x$.  If $x\in\Sigma^n$ then $x_i$, $1\le i\le n$, denotes
the $i$-th letter of $x$ and $x_A$, $A=\{i_1,\ldots,i_m\}, 1\le
i_1<\ldots,<i_m\le n$ denotes the word $x_{i_1}\ldots x_{i_m}$. We say
that a 0-1 sequence $\overline{x}=(x_1,\ldots,x_{N_k})$ is {\em
  2-sorted} if both $(x_1,x_3,\ldots,x_{N_k-1})$ and
$(x_2,x_4,\ldots,x_{N_k})$ are sorted.

\subsection{Reduction to Analysis of Columns}

For any $k\ge 3$ let $n_k=2^{k-1}-1$, $b_k=2(k-2)$ (thus $N_k=n_k\cdot
b_k$). The set of registers $Reg_k=\{1,\ldots,N_k\}$ can be analyzed as an
$n_k\times b_k$ matrix with $C_j^k=\{j+ib_k: 0\le i<n_k\}$,
$j=1,\ldots,b_k$, as columns. A content of all registers in the matrix,
that is $x\in\Sigma^{N_k}$, can be equivalently represented by the
sequence of contents of registers in $C_1$, $C_2$, \ldots, $C_{b_k}$,
that is $(x_{C_1},\ldots,x_{C_{b_k}})$.  Since $b_k$ is an even
number, the following fact is obviously true.
\begin{fact} If $x\in\Sigma^{N_k}$ is 2-sorted then each $x_{C_j}$,
  $j=1,\ldots,b_k$, is sorted. \qed
\end{fact}

That is, the columns are sorted at the beginning of a computation of
$2k-5$ passes of $M_k$. The first lemma we would like to prove is that
columns remain sorted after each stage of the computation. We start
with a following technical fact:

\begin{fact}\label{f33}
Let $A=\{a_1,\ldots,a_n\}$ and $B=\{b_1,\ldots,b_n\}$ be subsets of
$\{1,\ldots,N_k\}$ such that $a_1 < b_1 < a_2 < b_2 < \ldots < a_n <
b_n$. Let $h\ge 0$ and $S_{A,B,h}=\{[a_i:b_{i+h}]: 1\le i \le
n-h\}$. Then for any $x\in\Sigma^{N_k}$ such that $x_A$ and $x_B$
are sorted, the output $y=S_{A,B,h}(x)$ has the following properties:
\begin{enumerate}
\item[(i)] $y_A$ and $y_B$ are sorted.
\item[(ii)] Let $m_1=ones(x_A)$ and $m_2=ones(x_B)$. Then $ones(y_A) =
  \min(m_1,m_2+h)$ and $ones(y_B) = \max(m_1-h,m_2)$.
\end{enumerate}
\end{fact}
\begin{proof}
To prove (i) we show only that $y_{a_i}\le y_{a_{i+1}}$ for
$i=1,\ldots,n-1$. If $1\le i<n-h$ then $y_{a_i} =
\min(x_{a_i},x_{b_{i+h}}) \le \min(x_{a_{i+1}},x_{b_{i+h+1}}) =
y_{a_{i+1}}$ since $\min$ is a non-decreasing function and both $x_A$
and $x_B$ are sorted . If $i=n-h$ then $y_{a_i} =
\min(x_{a_i},x_{b_{i+h}}) \le x_{a_{i+1}} = y_{a_{i+1}}$. For $i>n-h$
we have $y_{a_i} = x_{a_i} \le x_{a_{i+1}} = y_{a_{i+1}}$.

To prove (ii) let $m'_1=\min(m_1,m_2+h)$ and $m'_2=\max(m_1-h,m_2)$.
We consider two cases. If $m_1\le m_2+h$ then $m_1-h\le m_2$ and we
get $m'_1=m_1$ and $m'_2=m_2$. In this case no comparator from
$S_{A,B,h}$ exchanges 0 with 1. To see this assume a.c. that a
comparator $[a_i:b_{i+h}]$ exchanges $x_{a_i}=1$ with
$x_{b_{i+h}}=0$. Then $i>n-m_1$ and $i+h\le n-m_2$ hold because of the
definitions of $m_1$ and $m_2$. It follows that $n-m_1 < n-m_2-h$,
thus $m_1-h > m_2$ --- a contradiction. If $m_1 > m_2+h$ then $m'_1 =
m_2+h$ and $m'_2=m_1-h$. In this case let us observe that a comparator
$[a_i:b_{i+h}]$ exchanges $x_{a_i}=1$ with $x_{b_{i+h}}=0$ if and only
if $m_2+h \le n-i < m_1$. Therefore $ones(y_A) = m_1-(m_1-m_2-h) =
m_2+h$ and $ones(y_B) = m_2+(m_1-m_2-h) = m_1-h$. \qed
\end{proof}
According to the definition of $M_k$, it consists of three stages $T^k_1,
T^k_2, T^k_3$, where $T^k_i = \cup\{S_{k,i+3j}: 0 \le j\le \lfloor
\frac{2k-i-3}{3} \rfloor\}$ (sets $S_j$ are defined in Def.~\ref{defMk}). 
Using the notation from Fact \ref{f33}, the following fact is an easy
consequence of Definition \ref{defMk}.
\begin{fact}\label{f34} Let $L_i=C_i$ and $R_i=C_{b_k-i+1}$ denote the
  corresponding left and the right columns of registers, and
  $h_i=2^{k-i-1}-1$, $i=1,\ldots,\frac{b_k}{2}$. Then
\begin{itemize}
\item[(i)] $regs(S_{k,1})\subseteq L_1\cup R_1$ and $S_{k,1} =
  S_{R_1-\{N_k\},L_1-\{1\},0}$
\item[(ii)] $regs(S_{k,2j})\subseteq L_j\cup R_j$ and
  $S_{k,2j} = S_{L_j,R_j,h_j}$, for any $j=1,\ldots,\frac{b_k}{2}$
\item[(iii)] $regs(S_{k,2j+1})\subseteq L_j\cup L_{j+1}\cup R_{j+1}\cup
  R_j$ and $S_{2j+1} = S_{L_j,L_{j+1},0}\cup S_{R_{j+1},R_j,0}$, for any
  $j=1,\ldots,\frac{b_k}{2}-1$
\item[(iv)] $regs(S_{k,2k-3})\subseteq L_{k-2}\cup R_{k-2}$ and
  $S_{k,2k-3} = S_{L_{k-2},R_{k-2},0}$
\item[(v)] if $(L_j\cup R_j)\cap regs(S_{k,i})\neq\emptyset$ then $2j-1\le
  i\le 2j+1$, for any $j=1,\ldots,\frac{b_k}{2}-1$
\end{itemize} \qed
\end{fact}
\begin{lemma} \label{l35} If the initial content of registers is a
  2-sorted 0-1 sequence $x$ then after each stage of multi-pass
  computation of $M_k=T^k_1,T^k_2,T^k_3$ the content of each column $C_j$,
  $j=1,\ldots,b_k$, is sorted, that is, each $(x^{(p,i)})_{C_j}$ is of
  the form $0^*1^*$, $p=0,\ldots$, $i=1,2,3$.
\end{lemma}
\begin{proof} By induction it suffices to prove that for each sequence
  $y\in\Sigma^{N_k}$ with sorted columns $C_j$, $j=1,\ldots,b_k$, the
  outputs $z_i=T^k_i(y)$, $i=1,2,3$ have also the columns sorted. Since
  each $T^k_i$, as a mapping, is a composition of mapping $S_{i+3j}, 0 \le
  j\le \lfloor \frac{2k-i-3}{3} \rfloor$, each of which, due to Facts 
  \ref{f33} and \ref{f34}, transforms sorted columns into sorted columns, the
  lemma follows. \qed
\end{proof}
From now on, instead of looking at 0-1 sequences with sorted columns, we
will analyze the computations of $M_k$ on sequences of integers
$\overline{c}=(c_1,\ldots,c_{b_k})$, where $c_t$, $t=1,\ldots,b_k$,
denote the number of ones in a sorted column $C_t$. Transformations of
0-1 sequences defined by sets $S_j$, $j=1,\ldots,2k-3$ will be
represented by the following mappings:
\begin{definition} \label{defFun}
Let $k\ge 3$, $h_i=2^{k-i-1}-1$ for $i=1,\ldots,k-2$ and
$b_k=2(k-2)$. The functions $dec^k_i$, $mov^k_i$ and $cyc^k$ over
sequences of $b_k$ reals are defined as follows. Let
$\overline{c}=(c_1,\ldots,c_{b_k})$ and $t\in\{1,\ldots,b_k\}$.
\begin{eqnarray}
(dec^k_i(\overline{c}))_t &=& \left\{
\begin{array}{ll}
\min(c_i,c_{b_k-i+1}+h_i) & \mbox{ if } t=i\\
\max(c_i-h_i,c_{b_k-i+1}) & \mbox{ if } t=b_k-i+1\\
c_t                       & \mbox{ otherwise} 
\end{array}
\right.\\
(mov^k_i(\overline{c}))_t &=& \left\{
\begin{array}{ll}
\min(c_t,c_{t+1}) & \mbox{ if } t=i \mbox{ or } t=b_k-i\\
\max(c_{t-1},c_t) & \mbox{ if } t=i+1 \mbox{ or } t=b_k-i+1\\
c_t                       & \mbox{ otherwise} 
\end{array}
\right.\\
(cyc^k(\overline{c}))_t &=& \left\{
\begin{array}{ll}
\max(c_1,c_{b_k}-1) & \mbox{ if } t=1\\
\min(c_1+1,c_{b_k}) & \mbox{ if } t=b_k\\
c_t                       & \mbox{ otherwise} 
\end{array}
\right.
\end{eqnarray}
\end{definition}
\begin{fact} \label{f6}
Let $x\in\Sigma^{N_k}$ be a 0-1 sequence with sorted columns $C_1, \ldots,
  C_{b_k}$, let $c_i=ones(x_{C_i})$ and $\overline{c} =
  (c_1,\ldots,c_{b_k})$. Let $y_j=S_{k,j}(x)$, $d_{j,i}=ones((y_j)_{C_i})$
  and $\overline{d_j} =(d_{j,1},\ldots,d_{j,b_k})$, where
  $i=1,\ldots,b_k$ and $j=1,\ldots,2k-3$. Then
\begin{enumerate}
\item[(i)] $\overline{d_1} = cyc^k(\overline{c})$
\item[(ii)] $\overline{d_{2j}} = dec^k_j(\overline{c})$, for any
  $j=1,\ldots,\frac{b_k}{2}$
\item[(iii)] $\overline{d_{2j+1}} = mov^k_j(\overline{c})$, for any
  $j=1,\ldots,\frac{b_k}{2}$
\end{enumerate}
\end{fact}
\begin{proof} Generally, the fact follows from Fact \ref{f34} and the
  part (ii) of Fact \ref{f33} We prove only its parts (i) and (ii). Part
  (iii) can be proved in the similar way.

\textit{(i)~} Observe that $y_1=S_{k,1}(x)=S_{R_1-\{N_k\},L_1-\{1\},0}(x)$
  due to Fact \ref{f34} \textit{(ii)}. It follows that only the content
  of columns $L_1=C_1$ and $R_1=C_{b_k}$ can change, but they remain
  sorted (according to Lemma \ref{l35}). Using Fact
  \ref{f33} \textit{(ii)} we have: $m_1 = ones(x_{R_1-\{N_k\}}) =
  c_{b_k}-x_{N_k}$, $m_2 = ones(x_{L_1-\{1\}}) = c_1-x_1$ and
$$d_{1,1} = \max(m_1,m_2)+x_1 = \max(c_{b_k}-x_{N_k}+x_1,c_1),$$
$$d_{1,b_k} = \min(m_1,m_2)+x_{N_k} = \min(c_{b_k},c_1+x_{N_k}-x_1).$$
Now let us consider the following three cases of values $x_1$ and
$x_{N_k}$: \\
\noindent
{\it Case $x_1=0$ ~and~ $x_{N_k}=1$.} Then 
  $d_{1,1} = \max(c_{b_k}-1,c_1) = cyc^k(\overline{c})_1$ and 
  $d_{1,b_k} = \min(c_{b_k},c_1+1)=cyc^k(\overline{c})_1$.\\
\noindent
{\it Case $x_1=1$.} Then $c_1=n_k$, $c_{b_k}\le n_k$ and
$c_{b_k}-x_{N_k}\le n_k-1$. In this case:
  $d_{1,1} = \max(n_k,c_{b_k}-x_{N_k}+1,) = n_k = \max(c_1,c_{b_k}-1)$ and 
  $d_{1,b_k} = \min(n_k-1+x_{N_k},c_{b_k}) = c_{b_k} = \min(c_1+1,c_{b_k})$.\\
\noindent
{\it Case $x_{N_k}=0$.} Then $c_{b_k}=0$ and $c_1-x_1\ge 0$. In this case:
  $d_{1,1} = \max(c_1,x_1) = c_1 = \max(c_1,c_{b_k}-1)$ and 
  $d_{1,b_k} = \min(c_1-x_1,c_{b_k}) = c_{b_k} = \min(c_1+1,c_{b_k})$.

\textit{(ii)~} We fix any $j\in\{1,\ldots,\frac{b_k}{2}\}$ and observe
  that $y_{2j}=S_{2j}(x)=S_{L_j,R_j,h_j}(x)$ due to Fact \ref{f34}
  \textit{(ii)}. It follows that only the content of columns $L_j=c_j$
  and $R_j=c_{b_k-j+1}$ can change, but they remain sorted (according to
  Lemma \ref{l35}). Using Fact \ref{f33} \textit{(ii)} we have:
$$d_{2j,j} = ones((y_{2j})_{L_j}) = \min(c_j,c_{b_k-j+1}+h_j) = 
  (dec^k_j(\overline{c}))_j,$$
$$d_{2j,b_k-j+1} = ones((y_{2j})_{R_j}) = \max(c_j-h_j,c_{b_k-j+1}) = 
  (dec^k_j(\overline{c}))_{b_k-j+1}.$$
\end{proof}
\begin{definition} \label{defQ}
Let $k\ge 3$. Let $Q^k_1$, $Q^k_2$ and $Q^k_3$ denote the following sets
of functions.
\begin{eqnarray}
Q^k_1 & = & \left\{cyc^k\right\} \cup 
          \left\{dec^k_{3i-1}\right\}_{i=1}^{\lfloor\frac{k-1}{3}\rfloor} \cup
          \left\{mov^k_{3i}\right\}_{i=1}^{\lfloor\frac{k-2}{3}\rfloor}\\
Q^k_2 & = & \left\{dec^k_{3i-2}\right\}_{i=1}^{\lfloor\frac{k}{3}\rfloor} \cup
          \left\{mov^k_{3i-1}\right\}_{i=1}^{\lfloor\frac{k-1}{3}\rfloor}\\
Q^k_3 & = & \left\{dec^k_{3i}\right\}_{i=1}^{\lfloor\frac{k-2}{3}\rfloor} \cup
          \left\{mov^k_{3i-2}\right\}_{i=1}^{\lfloor\frac{k}{3}\rfloor}
\end{eqnarray}\
\end{definition}
Let us observe that each function in $Q^k_i$, $i=1,2,3$, can only modify
a few positions in a given sequence of numbers. Moreover, different
functions in $Q^k_i$ can only modify disjoint sets of positions. For a
function $f:R^m\mapsto R^m$ let us
define $$args(f)=\left\{i\in\{1,\ldots,m\}:\exists_{\overline{c}\in R^m}
(f(\overline{c}))_i\neq (\overline{c})_i\right\}$$ The following facts
formalize our observations.
\begin{fact}
$args(cyc^k)=\{1,b_k\}$, $args(dec^k_i)=\{i,b_k-i+1\}$, 
    $args(mov^k_i)=\{i,i+1,b_k-i,b_k-i+1\}$, where $i=1,\ldots,k-2$.
\end{fact} \qed
\begin{fact}
For each pair of functions $f,g\in Q^k_i$, $f\neq g$, $i=1,2,3$, we have
\begin{itemize}
\item[(i)] $args(f) \cap args(g) = \emptyset$;
\item[(ii)] for any $\overline{c}=(c_1,\ldots,c_{b_k})$ and 
                                                  $j\in\{1,\ldots,b_k\}$
\begin{equation}
(f(g(\overline{c})))_j = \left\{ \begin{array}{ll}
(f(\overline{c}))_j & \mbox{ if } j\in args(f)\\
(g(\overline{c}))_j & \mbox{ if } j\in args(g)\\
c_j                 & \mbox{ otherwise }
\end{array}
\right.
\end{equation}
\end{itemize} 
\end{fact}  \qed
\begin{corollary} \label{color1}
Each set $Q^k_i$, $i=1,2,3$, uniquely determines a mapping, in which
functions from $Q^k_i$ can be apply in any order. Moreover, if $f\in
Q^k_i$, $\overline{c}\in R^{b_k}$ and $j\in args(f)$ then
$(Q^k_i(\overline{c}))_j = (f(\overline{c}))_j$.
\end{corollary}
We would like to prove that the result of applying $Q^k_i$, $i=1,2,3$,
to a sequence $\overline{c}=(c_1,\ldots,c_{b_k})$ of numbers of ones in
columns $C_1,\ldots,C_{b_k}$ is equivalent to applying the set of
comparators $T^k_i$ to the content of registers, if each column is
sorted.
\begin{lemma} \label{l2}
Let $x\in\Sigma^{N_k}$ be a 0-1 sequence with sorted columns $C_1, \ldots,
  C_{b_k}$, let $c_i=ones(x_{C_i})$ and $\overline{c} =
  (c_1,\ldots,c_{b_k})$. Let $y_j=T^k_j(x)$, $d_{j,i}=ones((y_j)_{C_i})$
  and $\overline{d_j} =(d_{j,1},\ldots,d_{j,b_k})$, where
  $i=1,\ldots,b_k$ and $j=1,2,3$. Then $Q^k_j(\overline{c}) = \overline{d_j}$.
\end{lemma}
\begin{proof}
Recall that $T^k_j=\bigcup\{S_{k,j+3i}: 0\le i\le
\frac{2k-j-3}{3}\}$. For a set of comparators $S$ let us define
$$ cols(S) = \left\{ i\in\{1,\ldots,b_k\}: regs(S)\cap C_i\neq\emptyset
\right\} \enspace .$$ From Fact \ref{f34}(i--iv) it follows that
$cols(S_{k,1})=\{1,b_k\}$ and for $i=1,\ldots,k-2$ $cols(S_{k,2i}) =
\{i,b_k-i+1\}$ and $ cols(S_{k,2i+1}) = \{i, i+1, b_k-i, b_k-i+1\}
$. From Fact \ref{f34}(v) we get that $ cols(S_{k,j+3i}) \cap
cols(S_{k,j+3i'}) = \emptyset $ if $i\neq i'$. Thus we can observe a 1-1
correspondence between a function $f$ in $Q^k_j$ and a set of
comparators $S_{k,j+3i} \subseteq T^k_j$ such that
$args(f)=cols(S_{k,j+3i})$ Then for each $t\in args(f)$
$(Q^k_j(\overline{c}))_t = (f(\overline{c}))_t = (\overline{d_j})_t$, as
the consequence of Corollary \ref{color1} and Fact \ref{f6}. \qed
\end{proof}
\begin{definition} \label{flat}
We say that a sequence of numbers $\overline{c} = (c_1,\ldots,c_{2m})$
is {\em flat} if $c_1\le c_2\le \ldots, c_{2m}\le c_1 + 1$. We say that
a sequence $\overline{c}$ is {\em 2-flat} if subsequences
$(c_1,c_3,\ldots,c_{2m-1})$ and $(c_2,c_4,\ldots,c_{2m})$ are flat. We
say that $\overline{c}$ is balanced if $c_i+c_{2m-i+1} = c_1+c_{2m}$,
for $i=2,\ldots,m$. For a balanced sequence $\overline{c}$ define 
$height(\overline{c})$ as $c_1+c_{2m}$.
\end{definition}
\begin{proposition} \label{p3}
Let $k\ge 3$, $x\in \Sigma^{N_k}$, $\overline{c} =
(c_1,\ldots,c_{b_k})$, where $c_i=ones(x_{C_i})$ ($C_i$ is as usual a
column in the matrix of registers), $i=1,\ldots,b_k$. Then
\begin{enumerate}
\item $x$ is sorted if and only if columns of $x$ are sorted  and 
$\overline{c}$ is flat;
\item $x$ is 2-sorted if and only if columns of $x$ are sorted and 
$\overline{c}$ is 2-flat;
\end{enumerate}
\end{proposition} \qed

Now we are ready to reduce the proof of Theorem \ref{3merger} to the
proof of following lemma.
\begin{lemma} \label{l3}
Let $k\ge 3$. If for each 2-flat sequence $\overline{c} =
(c_1,\ldots,c_{b_k})$ of integers from $[0,2^{k-1}-1]$ the result of
application $(Q^k_3\circ Q^k_2\circ Q^k_1)^{2k-5}$ to $(\overline{c})$ is a
flat sequence, then $M_k$ is a $2k-5$-pass merger of two sorted sequences
given in odd and even registers, respectively.
\end{lemma}
\begin{proof}
Assume that for each 2-flat sequence $\overline{c} =
(c_1,\ldots,c_{b_k})$ the result of application $(Q^k_3\circ Q^k_2\circ
Q^k_1)^{2k-5}$ to $(\overline{c})$ is a flat sequence.  Let
$\overline{x}\in \Sigma^{N_k}$ be a 2-sorted sequence and $\overline{c}
= (c_1,\ldots,c_{b_k})$, where $c_i=ones(\overline{x}_{C_i})$ ($C_i$ is
as usual a column in the matrix of registers), $i=1,\ldots,b_k$. Then
$\overline{c}$ is 2-flat due to Proposition \ref{p3} and each 
$c_i\in[0,2^{k-1}-1]$, because the height of columns is $2^{k-1}-1$. Recall 
that $\overline{x}^{(j)}=(M_k)^j(\overline{x})$ and let
$c_{j,i}=ones(\overline{x}^{(j)}_{C_i})$. Using Lemma \ref{l2} and easy
induction we get that the equality $(Q^k_3\circ Q^k_2\circ
Q^k_1)^{j}(\overline{c}) = (c_{j,1},\ldots,c_{j,b_k})$ is true for
$j=1,\ldots,2k-5$. Since $(Q^k_3\circ Q^k_2\circ
Q^k_1)^{2k-5}(\overline{c})$ is a flat sequence, the sequence
$\overline{x}^{(2k-5)}$ is sorted. \qed
\end{proof}

\subsection{Analysis of Balanced Columns}

Due to Lemma \ref{l3} we can only analyze the results of periodic
application of the functions $Q^k_1$, $Q^k_2$ and $Q^k_3$ to a sequence
of integers representing the numbers of ones in each register column. We
know also that an initial sequence is 2-flat. To simplify our analysis
further, we start it with initial values restricted to be balanced
2-flat sequences. Then we observe that the functions are monotone and
any 2-flat sequence can be bounded from below and above by balanced
2-flat sequences whose heights differ only by one.

\begin{lemma} \label{l4}
Let $k\ge 3$ and $\overline{c} = (c_1,\ldots,c_{b_k})$ be a balanced
sequence of numbers. Let $s = height(\overline{c})$ and let $f$ be a
function from $Q^k_1\cup Q^k_2\cup Q^k_3$. Then $f(\overline{c})$ is
also balanced and $height(f(\overline{c}))=s$.
\end{lemma}

\begin{proof}
Let $\overline{c}$ and $s$ be as in Lemma and let $f(\overline{c}) =
(d_1,\ldots,d_{b_k})$.  The function $f\in Q^k_1\cup Q^k_2\cup Q^k_3$
can be either $cyc^k$ or one of $mov^k_j$, $dec^k_j$, $j=1,\ldots,k-2$,
according to Definition \ref{defQ}. Each of the functions can only
modify one or two pairs of positions of the form $(i,b_k-i+1)$ in
$\overline{c}$ (see Definition \ref{defFun}). The other pairs are left
untouched, so the sum of their values cannot change. In case of $cyc^k$
the modified pair is $(1,b_k)$ and $d_1+d_{b_k} = \max(c_1,c_{b_k}-1) +
\min(c_1+1,c_{b_k}) = s$. In case of $dec^k_j$ the pair is $(j,b_k-j+1)$
and $ d_j+d_{b_k-j+1} = \min(c_j,c_{b_k-j+1}+h_j) +
\max(c_j-h_j,c_{b_k-j+1}) = \min(c_j-h_j,c_{b_k-j+1}) + h_j +
\max(c_j-h_j,c_{b_k-j+1}) = s$.  Finally, if $f=mov^k_j$ then we have
two pairs $(j,b_k-j+1)$ and $(j+1,b_k-j)$. Then $ d_j+d_{b_k-j+1} =
\min(c_j,c_{j+1}) + \max(c_{b_k-j},c_{b_k-j+1}) = \min(c_j,c_{j+1}) +
\max(s-c_{j+1},s-c_j) = s$ and in case of the second pair
$d_{j+1}+d_{b_k-j} = \max(c_j,c_{j+1}) + \min(c_{b_k-j}, c_{b_k-j+1}) =
\max(c_j,c_{j+1}) + \min(s-c_{j+1},s-c_j) = s$. \qed
\end{proof}

It follows from Lemma \ref{l4} that if we start the periodical
application of the functions $Q^k_1$, $Q^k_2$ and $Q^k_3$ to a balanced
2-flat initial sequence then it remains balanced after each function
application and its height will not changed. Therefore, we can only
trace the values in the first half of generated sequences. If needed, a
value in the second half can be computed from the height and the
corresponding value in the first half. To get a better view on the
structure of generated sequences, we subtract half of the height from
each element of the initial sequence and proceed with such modified
sequences to the end. At the end the subtracted value is added to each
element of the final sequence. The following fact justifies the described
above procedure.

\begin{fact}\label{fct-10}
Let $f$ be a function from $Q^k_1\cup Q^k_2\cup Q^k_3$. Then $f$ is
monotone and for each $t\in R$ and $(c_1,\ldots,c_{b_k})$ the following
equation is true
$$ f(c_1-t,\ldots,c_{b_k}-t) =  f(c_1,\ldots,c_{b_k}) - (t,\ldots,t) \enspace .$$
\end{fact}

\begin{proof}
The fact follows from the similar properties of $\min$ and $\max$
functions: they are monotone and the equations: $\min(x-t,y-t) =
\min(x,y)-t$ and $\max(x-t,y-t) = \max(x,y)-t$ are obviously true. Each
$f$ in $Q^k_1\cup Q^k_2\cup Q^k_3$ is defined with the help of these
simple functions, thus $f$ inherits the properties. \qed
\end{proof}

\begin{corollary}
Let $f = f_l\circ f_{l-1}\circ\ldots\circ f_1$, where $f_i\in\{Q^k_1,
Q^k_2, Q^k_3\}$, $1\le i\le l$. Then $f$ is monotone and for any $t\in
R$ and $(c_1,\ldots,c_{b_k})\in R^{b_k}$
\[ f(c_1-t,\ldots,c_{b_k}-t) =  
             f(c_1,\ldots,c_{b_k}) - (t,\ldots,t) \enspace .\] \qed
\end{corollary} 

\begin{definition} \label{reduce}
Let $\overline{c}=(c_1,\ldots,c_{b_k})\in R^{b_k}$ be a balanced
sequence and $s=height(\overline{c})$. We call $(c_1-\frac{s}{2},
c_2-\frac{s}{2}, \ldots, c_{k-2}-\frac{s}{2})\in R^{b_k/2}$ the reduced
sequence of $\overline{c}$ and denote it by $reduce(\overline{c})$. For
a sequence $\overline{d} = (d_1, \ldots, d_{k-2})\in R^{k-2}$ we define
$s$-extended sequence $ext(\overline{d},s)$ as
\[(d_1+\frac{s}{2}, d_2+\frac{s}{2}, \ldots, d_{k-2}+\frac{s}{2}, 
  \frac{s}{2}-d_{k-2}, \frac{s}{2}-d_{k-3}, \ldots, \frac{s}{2}-d_1) \enspace .\]
For any $t\in R$ and a function $f: R^{b_k}\mapsto R^{b_k}$ that maps
each balanced sequence to a balanced one and preserves its height let
$reduce(f,t)$ denote a function on $R^{k-2}$ such that
$(reduce(f,t))(\overline{d}) = reduce(f(ext(\overline{d},t)))$ for any
$\overline{d}\in R^{k-2}$.
\end{definition}

Observe that for a balanced sequence $\overline{c}$ with height $s$ the
sequence $ext(reduce(\overline{c}),s)$ is equal to
$\overline{c}$. Moreover, for any $t\in R$ and a sequence
$\overline{d}\in R_{k-2}$ the sequence $ext(\overline{d},t)$ is balanced
and its height is $t$, thus $reduce(ext(\overline{d},t)) =
\overline{d}$. Note also that functions $Q^k_1$, $Q^k_2$ and $Q^k_3$
preserve the property of being balanced and the sequence height (see
Lemma \ref{l4}), so we can analyze a periodical application of their
reduced forms to a reduced balanced 2-flat input.

\begin{fact}
Let $f = f_l\circ f_{l-1}\circ\ldots\circ f_1$, where $f_i\in\{Q^k_1,
Q^k_2, Q^k_3\}$, $1\le i\le l$. Let $\overline{c}\in R^{b_k}$ be
balanced and $s=height(\overline{c})$ Let $\hat{f_i} = reduce(f_i,s)$,
$1\le i\le l$. Then $f(\overline{c}) = ext((\hat{f_l} \circ
\hat{f_{l-1}} \circ \ldots \circ
\hat{f_1})(reduce(\overline{c})),s)$. \qed
\end{fact}

\begin{definition}\label{def-Cyc}
Let $MinMax(x,y) = (min(x,y),\max(x,y))$, $Min(x) = \min(x,-x)$, $Cyc(x)
= \max(x,-x-1)$ and $Dec_i(x) = \min(x,-x+H_i)$, where $H_i = 2^i-1,
i=1,\ldots$. Moreover, let us define the following sequences of
functions:
\begin{eqnarray}
\hat{Q}^k_1 & = & (Cyc) \oplus 
\bigoplus_{i=1}^{\lfloor\frac{k-3}{3}\rfloor} 
                                 (Dec_{k-3i}, MinMax) \oplus (F^k_1) \\
\hat{Q}^k_2 & = & \bigoplus_{i=1}^{\lfloor\frac{k-2}{3}\rfloor} 
                                 (Dec_{k-3i+1}, MinMax) \oplus (F^k_2) \\
\hat{Q}^k_3 & = & \bigoplus_{i=1}^{\lfloor\frac{k-2}{3}\rfloor} 
                                 (MinMax, Dec_{k-3i-1}) \oplus (F^k_3)\enspace,
\end{eqnarray}
where $\oplus$ denote concatenation of sequences and for $i=1,2$ \[
F^k_i = \left\{
\begin{array}{ll}
() & \mbox{ if } k\equiv 2i+1 \;(\mbox{mod } 3) \\
(Dec_1) & \mbox{ if } k\equiv 2i+2 \;(\mbox{mod } 3)\\
(Dec_2,Min) & \mbox{ if } k\equiv 2i \;(\mbox{mod } 3)\\
\end{array}
\right. \qquad
F^k_3 = \left\{
\begin{array}{ll}
() & \mbox{ if }  k\equiv 2 \;(\mbox{mod } 3)\\
(Min) & \mbox{ if }  k\equiv 0 \;(\mbox{mod } 3)\\
(MinMax) & \mbox{ if }  k\equiv 1 \;(\mbox{mod } 3)\\
\end{array}
\right.\]
\end{definition}

\begin{lemma}\label{l5}
Let $k\ge 3$ and $t\in R$. Then $reduce(Q^k_i,t) = \otimes \hat{Q}^k_i$, where
$i=1,2,3$ and $\otimes$ denotes the Cartesian product of a sequence of
functions.
\end{lemma}
\begin{proof}
Let $k\ge 3$, $i\in\{1,2,3\}$ and $t\in R$. Let $\overline{d}\in R^{k-2}$. By 
Def. \ref{reduce}, $(reduce(Q^k_i,t))(\overline{d}) = 
reduce(Q^k_i(ext(\overline{d},t)))$. Let $\overline{e} = ext(\overline{d},t) 
= (d_1+\frac{t}{2}, \ldots, d_{k-2}+\frac{t}{2}, -d_{k-2}+\frac{t}{2},\ldots, 
-d_1+\frac{t}{2})$. The sequence $\overline{e}$ is balanced and 
$height(\overline{e}) = t$. To get the lemma we would like to prove that for 
$j = 1, \ldots, k-2$ the equalities $(Q^k_i(\overline{e}))_j - \frac{t}{2} = 
((\otimes \hat{Q}^k_i)(\overline{d}))_j$ hold. The proof is by case analysis 
of 
values of $i$ and $j$. In the following equations we use Definitions 
\ref{defFun}, \ref{defQ}, \ref{reduce} and \ref{def-Cyc}.
\begin{enumerate}
\item (Case: $i=1$ and $j = 1$). Then $(Q^k_1(\overline{e}))_1 = 
(cyc^k(\overline{e}))_1 = \max(d_1+\frac{t}{2},-d_1+\frac{t}{2}-1) = 
\max(d_1,-d_1-1)+\frac{t}{2} = Cyc(d_1)+\frac{t}{2} = ((\otimes 
\hat{Q}^k_i)(\overline{d}))_1+\frac{t}{2}$.
\item (Case: $i+j > 2$ and $i+j\equiv 0 (mod~3)$). Let $l$ be such that $j = 
3l - i$. Then $(Q^k_i(\overline{e}))_j = (dec^k_{3l-i}(\overline{e}))_{3l-i} 
= \min(d_{3l-i}+\frac{t}{2},-d_{3l-i}+\frac{t}{2} + 2^{k-(3l-i)-1}-1) = 
\min(d_{3l-i},-d_{3l-i} + H_{k-(3l-i)-1}) + \frac{t}{2} = 
Dec_{k-3l+i-1}(d_{3l-i}) + \frac{t}{2} = ((\otimes 
\hat{Q}^k_i)(\overline{d}))_j + \frac{t}{2}$.
\item (Case: $i+j > 2$, $j<k-2$ and $i+j\equiv 1 (mod~3)$). Let $l$ be such 
that $j = 3l-i+1$. Then $(Q^k_i(\overline{e}))_j = 
(mov^k_{3l-i+1}(\overline{e}))_{3l-i+1} = 
\min(d_{3l-i+1}+\frac{t}{2},d_{3l-i+2}+\frac{t}{2}) = 
\min(d_{3l-i+1},d_{3l-i+2})+\frac{t}{2}$. Starting from the other side we get
$((\otimes \hat{Q}^k_i)(\overline{d}))_{3l-i+1} = (MinMax(d_{3l-i+1}, 
d_{3l-i+2}))_1 = \min(d_{3l-i+1},d_{3l-i+2})$ and we are done.
\item (Case: $i+j > 2$, $j=k-2$ and $i+j\equiv 1 (mod~3)$). Let $l$ be as in 
previous case. Then $(Q^k_i(\overline{e}))_{k-2} = 
(mov^k_{k-2}(\overline{e}))_{k-2} = 
\min(d_{k-2}+\frac{t}{2},-d_{k-2}+\frac{t}{2}) = \min(d_{k-2},-d_{k-2}) + 
\frac{t}{2} = Min(d_{k-2})+\frac{t}{2} = ((\otimes 
\hat{Q}^k_i)(\overline{d}))_{k-2} + \frac{t}{2}$.
\item (Case: $i+j > 2$ and $i+j\equiv 2 (mod~3)$). Let $l$ be such 
that $j = 3l-i+2$. Then $(Q^k_i(\overline{e}))_{3l-i+2} = 
(mov^k_{3l-i+1}(\overline{e}))_{3l-i+2} = 
\max(d_{3l-i+1}+\frac{t}{2}, d_{3l-i+2}+\frac{t}{2}) =$ \\ $\max(d_{3l-i+1}, 
d_{3l-i+2})+\frac{t}{2}$. Starting from the other side we get
$((\otimes \hat{Q}^k_i)(\overline{d}))_{3l-i+2} = (MinMax(d_{3l-i+1}, 
d_{3l-i+2}))_2 = \max(d_{3l-i+1},d_{3l-i+2})$ and we are finally done. \qed
\end{enumerate}
\end{proof}

Instead of tracing individual values in reduced sequences after each
application of a function from
$\{\otimes\hat{Q}^k_1,\otimes\hat{Q}^k_2,\otimes\hat{Q}^k_3\}$ we 
will
trace intervals in which the values should be and observe how the
lengths of intervals are decreasing during the computation. So let us
now define the intervals and show a fact about computations on them.
\begin{definition} \label{interval}
Let $k\ge 3$, $H_i=2^i-1$ for $1\le i\le k-1$. Let $I(0)$ denote the
interval $[-\frac 12,0]$ and, in similar way, let $I(i) = [-\frac
  12,\frac{H_i}{2}]$, $1\le i\le k-1$, $I(-k) = [-\frac{H_{k-1}}{2},0]$
and $I(\pm k) = [-\frac{H_{k-1}}{2},\frac{H_{k-1}}{2}]$. Moreover, we
will write $I(w_1,w_2,\ldots,w_l)$ for the Cartesian product $I(w_1)
\times I(w_2) \times \ldots \times I(w_l)$, where each $w_i \in
\{0,1,2,\ldots,k-1,-k,\pm k\}$.
\end{definition}
\begin{fact} \label{fct-12}
The following inclusions are true:
\begin{enumerate}
\item $Dec_i(I(i+1))\subseteq I(i)$ and 
  $Dec_i(I(w))\subseteq I(w)$, for $1\le i\le k-2$ and $w\in\{0,-k,\pm k\}$;
\item $Cyc(I(-k)) \subseteq I(k-1)$ and 
  $Cyc(w)\subseteq Cyc(w)$, for $w\in\{0,k-1\}$;
\item $Min(I(\pm k)) \subseteq I(-k)$ and $Min(I(1)) \subseteq I(0)$;
\item $MinMax(I(\pm k,-k)) \subseteq (I(-k,\pm k))$;
\item $MinMax(I(i,w)) \subseteq (I(w,i))$, 
                              for $1\le i\le k-1$ and $w\in\{0,-k\}$.
\end{enumerate}
\end{fact}
\begin{proof}
The proof of each inclusion is a straightforward consequence of the
definitions of a given function and intervals. Therefore we check only
inclusions given in the first item. Let $x\in I(i+1) = [-\frac
  12,\frac{H_{i+1}}{2}]$. If $x\in I(i) = [-\frac
  12,\frac{H_i}{2}]$. then $Dec_i(x) = \min(x,-x+H_i) = x$ since $2x\le
H_i$. Otherwise $x$ must be in $(\frac{H_i}{2},\frac{H_{i+1}}{2}]$, but
  then $x > -x+H_i$ and $Dec_i(x) = -x+H_i \in [-\frac
    12,\frac{H_i}{2})$ since $H_{i+1} = 2H_i+1$.

To proof the second inclusion for $Dec_i$ let us observe that if $x\le 0$
then $Dec_i(x) = x$. It follows that $Dec_i(I(0))\subseteq I(0)$ and
$Dec_i(I(-k))\subseteq I(-k)$. In case of $x\in I(\pm k)$ we only have
to check the positive values of $x$. such that $x\ge -x+H_i$. But then
$Dec_i(x) = -x+H_i > -x$ and both $x,-x\in I(\pm k)$. \qed
\end{proof}

Now we are ready to define sequences of intervals that are used to
describe states of computation after each periodic application of
functions $\hat{Q}^k_1$, $\hat{Q}^k_2$ and $\hat{Q}^k_3$ to a reduced
sequence of numbers of ones in columns.

\begin{definition}
Let $k\ge 3$. By $Z^k$ we denote the sequence
$(0,0,0)^{\lceil\frac{k-2}{3}\rceil}$ and, in the similar way, $U^k_1 =
(\pm k, \pm k, -k)^{\lceil\frac{k-2}{3}\rceil}$, $U^k_2 = (\pm k, -k,
\pm k)^{\lceil\frac{k-2}{3}\rceil}$ and $U^k_0 = (-k, \pm k, \pm
k)^{\lceil\frac{k-2}{3}\rceil}$.

Next, let $V^k_1 = \bigoplus_{i=1}^{\lceil\frac{k-2}{3}\rceil} (k-3i+2,
k-3i, -k)$, $V^k_2 = \bigoplus_{i=1}^{\lceil\frac{k-2}{3}\rceil}
(k-3i+1, -k, k-3i)$ and let $V^k_0 =
\bigoplus_{i=1}^{\lceil\frac{k-2}{3}\rceil} (-k, k-3i+1, k-3i-1)$.

Finally, let $W^k_1 = \bigoplus_{i=1}^{\lceil\frac{k-2}{3}\rceil}
(k-3i+2, k-3i, 0)$, $W^k_2 = \bigoplus_{i=1}^{\lceil\frac{k-2}{3}\rceil}
(k-3i+1, 0, k-3i)$ and let $W^k_0 =
\bigoplus_{i=1}^{\lceil\frac{k-2}{3}\rceil} (0, k-3i+1, k-3i-1)$.
\end{definition}

Note that all sequences defined above are of length
$3\lceil\frac{k-2}{3}\rceil \ge k-2$ and their elements are interval
descriptors as defined in Definition \ref{interval}.

\begin{definition}
Let $k\ge 3$. Let $\overline{a}=(a_1,\ldots,a_n)$ and
$\overline{b}=(b_1,\ldots,b_n)$ be any sequences, where $n\ge k-2$. For
$0\le i\le k-2$ let $join_k(i,\overline{a},\overline{b})$ denote $(a_1,
\ldots, a_i, b_{i+1}, \ldots, b_{k-2})$.
\end{definition}

\begin{definition}\label{def-X}
Let $k\ge 3$. Let $X^k_i$ denote a state sequence after $i$ stages and
be defined as:
$$ X^k_i = \left\{
\begin{array}{ll}
join_k(\lceil\frac{i+1}{2}\rceil,V^k_{i~mod~3},U^k_{i~mod~3}) & 
					\mbox{ for }i = 1,\ldots,2k-5\\
join_k(3k-6-i,V^k_{i~mod~3},W^k_{i~mod~3}) & 
					\mbox{ for }i = 2k-4,\ldots,3k-7\\
join_k(\lceil\frac{i+1-(3k-6)}{2}\rceil,Z^k,W^k_{i~mod~3}) & 
					\mbox{ for }i = 3k-6,\ldots,5k-12\\
\end{array}
\right. $$
\end{definition}

For example, to create $X^k_1$ we take the first element of $V^k_1$ and
the rest of elements from $U^k_1$ obtaining the sequence $(k-1,\pm
k,-k,\pm k,\pm k,-k,\pm k,\pm k, -k, \ldots)$ of length $k-2$. In the
next lemma we claim that $X^k_1$ really describes the state after the
first stage of computation, where input is a balanced 2-flat sequence.

\begin{lemma}\label{l6}
Let $k\ge 3$ and let $\overline{c} = (c_1,\ldots,c_{b_k})$ be a balanced
2-flat sequence of integers from $[0,2^{k-1}-1]$. Let $s =
height(\overline{c})$ and let $\overline{d} =
reduce(\overline{c})$. Then $(\otimes \hat{Q}^k_1)(\overline{d}) \in
I(X^k_1)$.
\end{lemma}
\begin{proof}
Recall that $H_i=2^i-1$. Let $\overline{d} = (d_1,\ldots,d_{k-2})$ By
Definitions \ref{flat} and \ref{reduce} $s=c_i+c_{b_k-i+1}$ and each
$d_i = c_i-\frac{s}{2} = \frac{c_i-c_{b_k-i+1}}{2}$. Observe that each
$d_i\in I(\pm k) = [-\frac{H_{k-1}}{2},\frac{H_{k-1}}{2}]$. It follows
from the following sequence of inequalities: $-\frac{H_{k-1}}{2} \le 
\frac{-c_{b_k-i+1}}{2} \le \frac{c_i-c_{b_k-i+1}}{2} \le \frac{c_i}{2} \le 
\frac{H_{k-1}}{2}$. Moreover, the sequence $\overline{d}$ is 2-flat, because 
$\overline{c}$ is 2-flat. That means that $d_1\le d_3\le d_5\le \ldots \le 
d_{k'} \le d_1+1$ and $d_2\le d_4\le d_6\le \ldots \le d_{k''} \le d_2+1$, 
where $k' = 2\lceil\frac{k-2}{2}\rceil-1$ and $k'' = 
2\lfloor\frac{k-2}{2}\rfloor$.
\begin{fact} \label{fct-13}
Either $-\frac{1}{2}\le d_1$ and $d_{k''}\le 0$ or $-\frac{1}{2}\le d_2$ 
and $d_{k'}\le 0$.
\end{fact}
To prove the fact we consider three cases of the value of $d_1$.

\noindent{\em Case $d_1\ge 0$:} In this case we only have to prove that 
$d_{k''}\le 0$. But it is true since $d_{k''}=\frac{c_{k''}-c_{b_k-k''+1}}{2} 
\le \frac{c_{b_k}-c_1}{2} = -d_1 \le 0$. The last inequality holds, because 
$\overline{c}$ is 2-flat and both $k''$ and $b_k$ are even.

\noindent{\em Case $d_1\le -1$:} Then $d_{k'} \le d_1+1 \le 0$. Thus we have 
only to prove that $d_2\ge -\frac{1}{2}$. Similar to the previous case, we 
observe that $d_2 = \frac{c_2-c_{b_k-1}}{2} \ge \frac{c_{b_k}-1-(c_1+1)}{2} = 
-d_1-1 \ge 0$.

\noindent{\em Case $d_1 = -\frac{1}{2}$:} Then $d_{k'} \le d_1+1 = \frac{1}{2}$
and from $-\frac{1}{2} = \frac{c_1-c_{b_k}}{2}$ we get $c_1+1 = c_{b_k} \le
c_2+1$. Since $c_2 \ge c_1$, we have $d_2 \ge d_1 = -\frac{1}{2}$. If $d_{k'}\le
0$, we are done. Otherwise $d_{k'} = \frac{1}{2}$ and we have to show that
$d_{k''} \le 0$. To this end let us notice that $\frac{s}{2} = c_1-d_1 =
c_1+\frac{1}{2}$ and $c_{b_k-k'+1} = s-c_{k'} = s-(d_{k'}+\frac{s}{2}) =
\frac{s}{2}-\frac{1}{2} = c_1$. It follows that $c_{k''} = c_1$ since $c_1 \le c_2
\le c_{k''} \le c_{b_k-k'+1} = c_1$. Thus $d_{k''} = d_1 = -\frac{1}{2}$ and this
concludes the proof of Fact \ref{fct-13}.

From Fact \ref{fct-13} and since $\overline{d}$ is 2-flat we can immediately 
get the following corollary.
\begin{corollary}
$\overline{d} \in I((k-1,-k,k-1,-k,\ldots) \cup I(-k,k-1,-k,k-1,\ldots)$.
\end{corollary}
To finish the proof of the lemma we need one more fact:
\begin{fact}
$(\otimes \hat{Q}^k_1)(I((k-1,-k,k-1,-k,\ldots) \cup 
I(-k,k-1,-k,k-1,\ldots)) 
\subseteq I(X^k_1)$.
\end{fact}
To prove this fact let us firstly represent $X^k_1$ in the same form as 
$\hat{Q}^k_1$ is.
\[X^k_1 = (k-1) \oplus \bigoplus_{i=1}^{\lfloor\frac{k-3}{3}\rfloor} (\pm 
k,-k,\pm k) \oplus (Y^k_1), \]
where $Y^k_1$ is empty if $k\equiv 0 \;(\mbox{mod } 3)$, $Y^k_1 = (\pm k)$ if 
$k\equiv 1 \;(\mbox{mod } 3)$ and $Y^k_1 = (\pm k,-k)$ if 
$k\equiv 2 \;(\mbox{mod } 3)$. Looking now at both representations we can see 
that the output of $Cyc$ function should be in $I(k-1)$, the output of each 
$Dec_i$ function should be in $I(\pm k)$ and the output of $MinMax$ should be 
in $I(-k)\times I(\pm k)$. If $Min$ function is used, then its output should 
be in $I(-k)$. The input to $Cyc$ is either from $I(k-1)$ or from $I(-k)$. In 
both cases we get desired output according to Fact \ref{fct-12}.2. In the 
similar way, the input to each $Dec_i$ function is either from $I(k-1) 
\subseteq I(\pm k)$ or from $I(-k) \subseteq I(\pm k)$. But $Dec_i(I(\pm k)) 
\subseteq I(\pm k)$ by Fact \ref{fct-12}.1. From Fact \ref{fct-12}.3 we have 
$Min(I(\pm k)) \subseteq I(-k)$. Finally, the input to $MinMax$ function is 
either from $I(k-1)\times I(-k)$ or from $I(-k)\times I(k-1)$. For this 
function the result follows from Fact \ref{fct-12}.4. \qed
\end{proof}

\begin{lemma}\label{l7}
For $k\ge 3$ and each $i=1,2,\ldots,5k-13$ the following inclusion holds:
\[ (\bigotimes \hat{Q}^k_{i~mod~3 +1})(I(X^k_i))\subseteq I(X^k_{i+1}). \]
\end{lemma}
\begin{proof}
We have to prove that for $k\ge 3$ and $x = 1,2,3$ the following inclusions 
are true: $(\bigotimes \hat{Q}^k_x)(I(X^k_{3j+x-1}))\subseteq 
I(X^k_{3j+x})$, where $j=1,2,\ldots\lfloor\frac{5k-13}{3}\rfloor$ for $x=1$ 
and  $j=0,1,2,\ldots\lfloor\frac{5k-12-x}{3}\rfloor$ for $x=2,3$. The 
sequences $\hat{Q}^k_x$, $x=1,2,3$, are built of functions $Cyc$, 
$Dec_*$, $MinMax$ and $Min$ introduced in Definition \ref{def-Cyc}. We 
consider these function one after another analysing which positions in state
sequences are modified by them and what values are in that positions before 
and after applying a function. In the following, we denote by $A_{i,j}$ the $j$-th 
element of a sequence $A_i$.

The function $Cyc$ is used only in the definition of $\hat{Q}^k_1$ and is
applied to position 1 of state sequences $I(X^k_{3j})$, where
$j=1,2,\ldots\lfloor\frac{5k-13}{3}\rfloor$. Thus it is enough to show the
inclusion $Cyc(I(X^k_{3j,1}))\subseteq I(X^k_{3j+1,1})$. By Definition
\ref{def-X} the argument of $Cyc\cdot I$ can be: $X^k_{3j,1} = V^k_{0,1} = -k$
for $3j\le 3k-9$ or $X^k_{3j,1} = W^k_{0,1} = 0$ for $3j=3k-6$ or $X^k_{3j,1}
= Z_1 = 0$ for $3j>3k-6$. The corresponding value of the next state sequence
is $X^k_{3j+1,1} = V^k_{1,1} = k-1$ for $3j+1\le 3k-8$ or $X^k_{3j+1,1} = Z_1
= 0$ for $3j+1\ge 3k-5$. Using Fact \ref{fct-12}, inclusions $Cyc(I(-k))
\subseteq I(k-1)$ and $Cyc(I(0)) \subseteq I(0)$ are true and we are done.

In the sequence $\hat{Q}^k_1$ we have several $Dec_{k-3l}$ functions, each
$Dec_{k-3l}$ is on the corresponding position $3l-1$ and it is applied to the
state sequence $I(X^k_{3j})$, where $l=1,\ldots,\lfloor\frac{k-1}{3} \rfloor$.
Similarly, in $\hat{Q}^k_2$ we have several $Dec_{k-3l+1}$ functions, each
$Dec_{k-3l+1}$ is on the corresponding position $3l-2$ and it is applied to the 
state sequence $I(X^k_{3j+1})$, where $l=1,\ldots,\lfloor\frac{k}{3} \rfloor$.
Finally, in $\hat{Q}^k_3$ we have $Dec_{k-3l-1}$ functions, each $Dec_{k-3l-1}$
is on the corresponding position $3l$ and it is applied to the state sequence
$I(X^k_{3j+2})$, where $l=1,\ldots,\lfloor\frac{k-2}{3}\rfloor$. Assuming that
$\hat{Q}^k_0$ also denotes $\hat{Q}^k_3$, we can rewrite our proof goal for that
functions as the following fact.
\begin{fact}
For $k\ge 3$ and $x=0,1,2$ the set $Dec_{k-3l+x-1}(I(X^k_{3j+x-1,3l-x}))$ is
a subset of $I(X^k_{3j+x,3l-x})$, where $l=1,\ldots,\lfloor\frac{k-2+x}{3}
\rfloor$, $j=1,2,\ldots\lfloor\frac{5k-12-x}{3}\rfloor$ for $x=0,1$ and
$j=0,1,2,\ldots\lfloor\frac{5k-14}{3}\rfloor$ for $x=2$.
\end{fact}
The sequences $X^k_*$ are defined with the help of sequences $U^k_*$, $V^k_*$,
$W^k_*$ and $Z_*$, therefore we prove the fact by considering all possible
cases in the following table. In it we assume that $U^k_{-1}=U^k_2$,
$V^k_{-1}=V^k_2$ and $W^k_{-1}=W^k_2$.
\setlength{\tabcolsep}{3pt}
\begin{center}
\begin{tabular}{||l|l||c|c||c||}
\hline\hline \multicolumn{1}{||c|}{Cases of} & 
      \multicolumn{1}{|c||}{Cases of} & Value of  & Value of  & Why  \\ 
      $s=X^k_{3j+x-1,3l-x}$  & $t=X^k_{3j+x,3l-x}$  & $s$  & $t$  & 
      $Dec_{k-3l+x-1}(I(s)) \subseteq I(t)$? \\ \hline
\hline $s=U^k_{x-1,3l-x}$ & $t=U^k_{x,3l-x}$ & $\pm k$  & $\pm k$  & 
      \multirow{5}{*}{Fact \ref{fct-12}.1}  \\ 
\cline{1-4} $s=V^k_{x-1,3l-x}$ & $t=V^k_{x,3l-x}$ & $k-3l+x$ & $k-3l+x-1$ &\\ 
\cline{1-4} $s=V^k_{x-1,3l-x}$ & $t=W^k_{x,3l-x}$ & $k-3l+x$ & $k-3l+x-1$ &\\ 
\cline{1-4} $s=W^k_{x-1,3l-x}$ & $t=W^k_{x,3l-x}$ & $k-3l+x$ & $k-3l+x-1$ &\\ 
\cline{1-4} $s=Z_{3l-x}$ & $t=Z_{3l-x}$ & 0 & 0 & \\ 
\hline\hline 
\end{tabular}
\end{center}
The two remaining cases: (1) $X^k_{3j+x-1,3l-x} = U^k_{x-1,3l-x}$ and 
$X^k_{3j+x,3l-x} = V^k_{x,3l-x}$ and (2) $X^k_{3j+x-1,3l-x} = W^k_{x-1,3l-x}$ 
and $X^k_{3j+x,3l-x} = Z_{3l-x}$ are not possible, because, otherwise, (1) 
$3j+x$ should be even and $\frac{3j+x}{2} = 3l-x-1$, which cannot hold for 
any integers $j$, $x$ and $l$; (2) $3j+x -(3k-6)$ should be even and 
$\frac{3j+x-(3k-6)}{2} = 3l-x-1$, which is not true for the same reason.

Now we consider the $Min$ function. It appears in the definition of 
$\hat{Q}^k_1$ ($\hat{Q}^k_2$ or $\hat{Q}^k_3$, respectively) on the position 
$k-2$ if $k~mod~3 = 2$ ($k~mod~3 = 1$ or $k~mod~3 = 0$, respectively). Thus, 
to prove the lemma, it suffices to show the following fact.
\begin{fact}
For $k\ge 3$ and $x=0,1,2$ the set $Min(I(X^k_{3j+x-1,k-2}))$ is a subset of
$I(X^k_{3j+x,k-2})$, where  $k-2 \equiv 1-x~(mod~3)$
$j=1,2,\ldots\lfloor\frac{5k-12-x}{3}\rfloor$ for $x=0,1$ and
$j=0,1,2,\ldots\lfloor\frac{5k-14}{3}\rfloor$ for $x=2$.
\end{fact}
As in the case of $Dec_*$ functions we prove the fact by considering all 
possible cases in the following table. In it we assume that $U^k_{-1}=U^k_2$,
$V^k_{-1}=V^k_2$ and $W^k_{-1}=W^k_2$.
\begin{center}
\begin{tabular}{||l|l||c|c||c||}
\hline\hline
      \multicolumn{1}{||c|}{Cases of} & \multicolumn{1}{|c||}{Cases of} & 
      Value of  & Value of  & Why  \\ 
      $s=X^k_{3j+x-1,k-2}$  & $t=X^k_{3j+x,k-2}$  & $s$  & $t$  & 
      $Min(I(s)) \subseteq I(t)$? \\ \hline
\hline\multirow{2}{*}{$s=U^k_{x-1,k-2}$} & $t=U^k_{x,k-2}$ & $\pm k$ & $-k$  &
      \multirow{6}{*}{Fact \ref{fct-12}.3}  \\ 
\cline{2-4} & $t=V^k_{x,k-2}$ & $\pm k$ &  $-k$ & \\ 
\cline{1-4} $s=V^k_{x-1,k-2}$ & $t=W^k_{x,k-2}$ & $1$ & $0$ & \\ 
\cline{1-4} \multirow{2}{*}{$s=W^k_{x-1,k-2}$} & $t=W^k_{x,k-2}$ & $1$ & $0$ 
& \\ 
\cline{2-4} & $t=Z_{x,k-2}$ & $1$ & $0$ & \\ 
\cline{1-4} $s=Z_{k-2}$ & $t=Z_{k-2}$ & 0 & $0$ & \\ 
\hline\hline 
\end{tabular}
\end{center}
The remaining case $X^k_{3j+x-1,k-2} = V^k_{x-1,k-2}$ and 
$X^k_{3j+x,k-2} = V^k_{x,k-2}$ is not possible, because, otherwise $3j+x-1 = 
2k-5$, that is, $2(k-2) = 3j+x$, but $k-2 \equiv 1-x~(mod~3)$ and in the 
consequence $x \equiv 2(1-x)~(mod~3)$ - contradiction.

The last function we have to consider is $MinMax$, which appears in the
definition of all $\hat{Q}^k_x$, $x=1,2,3$, functions. In $\hat{Q}^k_1$
($\hat{Q}^k_2$ and $\hat{Q}^k_3$, respectively) a copy of $MinMax$ is on
positions $(3,4), (6,7), \ldots$ ((2,3), (5,6), \ldots and (1,2), (4,5), 
\ldots, respectively). Thus, to prove the lemma, it suffices to show the
following fact.
\begin{fact}
For $k\ge 3$ and $x=1,2,3$ the set $MinMax(I(X^k_{3j+x-1,3l-x+1},
X^k_{3j+x-1,3l-x+2}))$ is a subset of $I(X^k_{3j+x,3l-x+1},
X^k_{3j+x,3l-x+2})$, where $l=1,\ldots,\lfloor\frac{k-4+x}{3}\rfloor$
$j=1,2,\ldots\lfloor\frac{5k-13}{3}\rfloor$ for $x=1$ and
$j=0,1,2,\ldots\lfloor\frac{5k-12-x}{3}\rfloor$ for $x=2,3$.
\end{fact}
As in the case of previous functions we prove the fact by considering all
possible cases in the following table. In it we assume that $U^k_3=U^k_0$,
$V^k_3=V^k_0$ and $W^k_3=W^k_0$. To reduce the size of the table we also use
the following shortcuts: $a=3j+x$, $b=3l-x+1$ and $y=k-3l+x-2$. Observe that
$2\le y\le k-2$, therefore $I(0)\subseteq I(y)\subseteq I(\pm k)$ and we can 
also apply Fact \ref{fct-12}.5.
\begin{center}
\begin{tabular}{|*{4}{|l}|*{4}{|c}||c||}
\hline\hline \multicolumn{2}{||c|}{Cases of $(s_1,s_2)$} & 
      \multicolumn{2}{|c||}{Cases of $(t_1,t_2)$} & 
      \multicolumn{2}{|c|}{Value of} & 
      \multicolumn{2}{|c||}{Value of}  & Why $I(t_1,t_2)\supseteq$ \\ 
      $X^k_{a-1,b-1}$  & $X^k_{a-1,b}$  & $X^k_{a,b-1}$  & $X^k_{a,b}$  & 
      $s_1$  & $s_2$  & $t_1$  & $t_2$ & $MinMax(I(s_1,s_2))$? \\ \hline
\hline \multirow{2}{*}{$U^k_{x-1,b-1}$} & \multirow{2}{*}{$U^k_{x-1,b}$} &
      $U^k_{x,b-1}$ & $U^k_{x,b}$ & 
    $\pm k$ & $-k$ & $-k$ & $\pm k$ & \multirow{2}{*}{Fact \ref{fct-12}.4} \\
\cline{3-8} & & $V^k_{x,b-1}$ & $U^k_{x,b}$ & 
      $\pm k$ & $-k$ & $-k$ & $\pm k$ &   \\ 
\hline \multirow{5}{*}{$V^k_{x-1,b-1}$} & \multirow{2}{*}{$U^k_{x-1,b}$} & 
\multirow{2}{*}{$V^k_{x,b-1}$} & $U^k_{x,b}$ & 
    $y$ & $-k$ & $-k$ & $\pm k$ &  \multirow{10}{*}{Fact \ref{fct-12}.5}  \\ 
\cline{4-8} &  &  & $V^k_{x,b}$ & $y$ & $-k$ & $-k$ & $y$ &  \\ 
\cline{2-8} & \multirow{2}{*}{$V^k_{x-1,b}$} &\multirow{2}{*}{$V^k_{x,b-1}$}  
      & $V^k_{x,b}$ & $y$ & $-k$ & $-k$ & $y$ &   \\ 
\cline{4-8} & & & $W^k_{x,b}$ & $y$ & $-k$ & $-k$ & $y$ &   \\ 
\cline{2-8} & $W^k_{x-1,b}$ & $W^k_{x,b-1}$ & $W^k_{x,b}$ & 
      $y$ & $0$ & $0$ & $y$ &   \\ 
\cline{1-8} \multirow{2}{*}{$W^k_{x-1,b-1}$} & \multirow{2}{*}{$W^k_{x-1,b}$} 
      & $W^k_{x,b-1}$ & $W^k_{x,b}$ & $y$ & $0$ & $0$ & $y$ &   \\ 
\cline{3-8} & & $Z_{b-1}$ & $W^k_{x,b}$ & $y$ & $0$ & $0$ & $y$ & \\ 
\cline{1-8} \multirow{3}{*}{$Z_{b-1}$} & \multirow{2}{*}{$W^k_{x-1,b}$} & 
      \multirow{2}{*}{$Z_{b-1}$} & $W^k_{x,b}$ & $0$ & $0$ & $0$ & $y$ &  \\ 
\cline{4-8} & & & $Z_{b}$ & $0$ & $0$ & $0$ & $0$ &  \\ 
\cline{2-8} & $Z_{b}$ & $Z_{b-1}$ & $Z_{b}$ & $0$ & $0$ & $0$ & $0$ &  \\ 
\hline\hline 
\end{tabular}
\end{center} \qed
\end{proof}

\begin{lemma}
Let $k\ge 3$ and let $\overline{c} = (c_1,\ldots,c_{b_k})$ be a balanced
2-flat sequence of integers from $[0,2^{k-1}-1]$ and let
$s=height(\overline{c})$. Let $f = f_{5k-12}\circ f_{5k-13}\circ \ldots\circ
f_1$, where $f_i = Q^k_{((i-1)~mod~3)+1}$, $i=1,\ldots,5k-12$. Then 
$f(\overline{c})
= (\frac{s}{2})^{b_k}$ if $s$ is even or $f(\overline{c}) =
(\frac{s-1}{2})^{k-2}\oplus(\frac{s+1}{2})^{k-2}$ otherwise.
\end{lemma}
\begin{proof}
Since each $f_i$ maps a balanced sequence to a balanced one, let $\hat{f}_i
= reduce(f_i,s) = \bigotimes\hat{Q}^k_{((i-1)~mod~3)+1}$, where the later
equality follows from Lemma \ref{l5}. Let also $\overline{d}_0 =
reduce(\overline{c})$ and let $\overline{d}_i = \hat{f}_i(\overline{d}_{i-1})$
for $i=1,\ldots,5k-12$. Then $\overline{d}_1\in I(X^k_1)$ by Lemma \ref{l6}
and for $i=2,\ldots,5k-12$ we get $\overline{d}_i\in I(X^k_i)$ by an easy
induction and Lemma \ref{l7}. Let $\mathbb{Z}$ denote as usual the set of 
integers. By $\mathbb{Z}_{\frac{1}{2}}$ we will denote the set 
$\{z+\frac{1}{2}| z\in\mathbb{Z}\}$. Looking at Definitions \ref{reduce} and 
\ref{def-Cyc} observe the 
following fact:
\begin{fact}
If $s$ is even then all elements of sequences $\overline{d}_i$, $i=0,\ldots, 
5k-12$, are integers. If $s$ is odd then  all elements of sequences 
$\overline{d}_i$, $i=0,\ldots, 5k-12$, are in $\mathbb{Z}_{\frac{1}{2}}$.
\end{fact}
Since $\overline{d}_{5k-12} \in I(X^k_{5k-12}) = I(0^{k-2})$ and $I(0)\cap 
\mathbb{Z} = \{0\}$ and $I(0)\cap \mathbb{Z}_{\frac{1}{2}} = 
\{\frac{1}{2}\}$, it follows that $\overline{d}_{5k-12} = 0^{k-2}$ if $s$ is 
even and  $\overline{d}_{5k-12} = \frac{1}{2}^{k-2}$, otherwise. Applying now 
the definition of $s$-extended sequence to $0^{k-2}$ and $\frac{1}{2}^{k-2}$ 
we get the desired conclusion of the lemma. \qed
\end{proof}

In this way, with respect to Lemma \ref{l3}, we have proved that the network 
$M_k$ is able to merge in $5k-12$ stages two sorted sequences given in odd 
and even registers, provided that the numbers of ones in our matrix columns 
form a balanced sequence. If the sequence is not balanced, $k-3$ additional 
stages are needed to get a sorted output.

\subsection{Analysis of General Columns}

In a general case we will use balanced sequences as lower and upper bounds on
the numbers of ones in our matrix columns and observe that $Q^k_1$, $Q^k_2$
and $Q^k_3$ are monotone functions (see Fact \ref{fct-10}).

\begin{definition}\label{def-14}
Let $k\ge 3$ and let $\overline{c} = (c_1,\ldots,c_{b_k})$ be a
2-flat sequence of integers from $[0,2^{k-1}-1]$ that is not balanced.
Since both $\overline{c}_{odd} =(c_1,\ldots,c_{b_k-1})$ and 
$\overline{c}_{evn} =(c_2,\ldots,c_{b_k})$ are 
flat sequences, let $i$ ($j$, respectively) be such that $c_{2i-1} < 
c_{2i+1}$ 
($c_{b_k-2j} < c_{b_k-2j+2}$, respectively) or let $i=k-2$ ($j=k-2$) if 
$\overline{c}_{odd}$ ($\overline{c}_{evn}$, respectively) is a constant 
sequence. The defined below sequences $\check{c}$ and $\hat{c}$ we will call 
lower and upper bounds of $\overline{c}$. If $i < j$ then for $l=1, \ldots, 
b_k$ \[
\check{c}_l = \left\{
\begin{array}{ll}
c_1       & \mbox{ if } l \mbox{ is odd and } l\le 2j-1\\
c_{b_k-1} & \mbox{ if } l \mbox{ is odd and } l\ge 2j+1\\
c_l       & \mbox{ if } l \mbox{ is even}\\
\end{array}
\right. \qquad
\hat{c}_l = \left\{
\begin{array}{ll}
c_{b_k-1} & \mbox{ if } l \mbox{ is odd}\\
c_{b_k}   & \mbox{ if } l \mbox{ is even}\\
\end{array}
\right.\] 
If $i > j$ then for $l=1, \ldots, b_k$ \[
\check{c}_l = \left\{
\begin{array}{ll}
c_1      & \mbox{ if } l \mbox{ is odd}\\
c_2      & \mbox{ if } l \mbox{ is even}\\
\end{array}
\right. \qquad
\hat{c}_l = \left\{
\begin{array}{ll}
c_l       & \mbox{ if } l \mbox{ is odd}\\
c_2       & \mbox{ if } l \mbox{ is even and } l\le b_k-2i\\
c_{b_k}   & \mbox{ if } l \mbox{ is even and } l >  b_k-2i\\
\end{array}
\right.\] 
\end{definition}
\begin{fact}\label{fct-19}
For $k\ge 3$ and any not balanced 2-flat sequence $\overline{c} =
(c_1,\ldots,c_{b_k})$ of integers from $[0,2^{k-1}-1]$ the sequences
$\check{c}$ and $\hat{c}$ are balanced, $height(\check{c})+1 =
height(\hat{c})$ and $\check{c}\le\overline{c}\le\hat{c}$.
\end{fact}
\begin{proof}
Let $i$ and $j$ be defined as in Definition \ref{def-14}. We will only 
consider the case $i<j$. The proof of the other case is similar. Directly 
from the definition we get that $\hat{c}$ is balanced. To see that 
$\check{c}$ is also balanced let us check for $l=1,\ldots,k-2$ whether the 
sum $\check{c}_{2l-1} + \check{c}_{b_k-2l+2}$ is constant.
\[\check{c}_{2l-1} + \check{c}_{b_k-2l+2} = \check{c}_{2l-1} + c_{b_k-2l+2} = 
\left\{
\begin{array}{ll}
c_1 + c_{b_k-2l+2} = c_1 + c_{b_k}  & \mbox{ if } l\le j\\
c_{b_k-1} + c_{b_k-2l+2} = c_{b_k-1} + c_2 & \mbox{ otherwise }\\
\end{array}
\right.\]
If $j = k-2$ there is no otherwise case and we are done. If $j < k-2$ then 
$c_{b_k}-c_2 = c_{b_k-1}-c_1 = 1$, because of the definition of $i$ and $j$ 
and we are also done. Moreover $height(\check{c}) + 1 = c_1 + c_{b_k} +1 = 
c_{b_k-1} + c_{b_k} = height(\hat{c})$. To prove that $\check{c}\le 
\overline{c}\le \hat{c}$ we consider even and odd indices. For even indices 
from the definition we have: $\check{c}_{2l} = c_{2l} \le c_{b_k} = 
\hat{c}_{2l}$. For odd indices $\hat{c}_{2l-1} = c_{b_k-1} \ge c_{2l-1} \ge 
c_1$. If $l\le j$ we are done, otherwise, $c_{2l-1} = c_{b_k-1} = 
\check{c}_{2l-1}$, because $\overline{c}_{odd}$ is flat. \qed
\end{proof}

\begin{theorem}\label{thm-19}
Let $k\ge 3$ and let $\overline{c} = (c_1,\ldots,c_{b_k})$ be a 2-flat
sequence of integers from $[0,2^{k-1}-1]$. Let $f = f_{6k-15}\circ
f_{6k-14}\circ \ldots\circ f_1$, where $f_i = Q^k_{((i-1)~mod~3)+1}$,
$i=1,\ldots,6k-15$. Then $f(\overline{c})$ is a flat sequence.
\end{theorem}
\begin{proof}
For a a 2-flat sequence $\overline{c}$ of integers from $[0,2^{k-1}-1]$ let 
$\check{c}$ and $\hat{c}$ be its balanced lower and upper bounds, as defined 
in Definition \ref{def-14}. Let $\overline{c}_0 = \overline{c}$, $\check{c}_0 
= \check{c}$, $\hat{c}_0 = \hat{c}$ and for $i = 1,\ldots,6k-15$ let us 
define $\overline{c}_i = f_i(\overline{c}_{i-1})$, $\check{c}_i = 
f_i(\check{c}_{i-1})$ and $\hat{c}_i = f_i(\hat{c}_{i-1})$. Observe that 
$\check{c}_i\le \overline{c}_i\le \hat{c}_i$, because of monotonicity of 
functions $Q^k_1$, $Q^k_2$, $Q^k_3$ and Fact \ref{fct-19}. To prove that 
$\overline{c}_{6k-15}$ is a flat sequence we need the following three 
technical facts.
\begin{fact}\label{fct-21}
Let $s = height(\check{c})$. If $s$ is even then $\overline{c}_{i,j} =
\frac{s}{2}$ and $\overline{c}_{i,b_k-j+1}\in\{\frac{s}{2},\frac{s}{2}+1\}$
for each $i = 3k-6, \ldots, 5k-12$ and $j = 1, \ldots,
\lceil\frac{i+1-(3k-6)}{2}\rceil$. If $s$ is odd then  $\overline{c}_{i,j}
\in\{\frac{s-1}{2},\frac{s+1}{2}\}$ and $\overline{c}_{i,b_k-j+1} =
\frac{s+1}{2}$ for each $i = 3k-6, \ldots, 5k-12$ and $j = 1, \ldots,
\lceil\frac{i+1-(3k-6)}{2}\rceil$. 
\end{fact}
\begin{proof}
Since both $\check{c}$ and $\hat{c}$ are balanced, we can consider reduced 
forms of them and use Lemmas \ref{l6} and \ref{l7}. For the given range of 
$i$'s values that means that \[
reduce(\check{c}_i), reduce(\hat{c}_i) \in I(X^k_i) = 
I(join_k(\lceil\frac{i+1-(3k-6)}{2}\rceil,Z^k,W^k_i)).\]
It follows that for a given range of $j$'s values 
$reduce(\check{c}_i)_j, reduce(\hat{c}_i)_j \in I(0) = [-\frac{1}{2},0]$. 
From Fact \ref{fct-19} we know that $height(\hat{c}) = s+1$ and from Lemma 
\ref{l4} that heights are preserved in sequences $\check{c}_i$ and 
$\hat{c}_i$. Thus, from the definition of a reduced sequence, 
$\check{c}_{i,j} \in [\frac{s-1}{2},\frac{s}{2}]$, $\check{c}_{i,b_k-j+1} \in 
[\frac{s}{2},\frac{s+1}{2}]$, $\hat{c}_{i,j} \in [\frac{s}{2},\frac{s+1}{2}]$
and $\hat{c}_{i,b_k-j+1} \in [\frac{s+1}{2},\frac{s+2}{2}]$. Since 
$\check{c}_i$ and $\hat{c}_i$ are sequences of integers, for even $s$ we get 
$\check{c}_{i,j} = \check{c}_{i,b_k-j+1} = \hat{c}_{i,j} = \frac{s}{2}$
and $\hat{c}_{i,b_k-j+1} = \frac{s+2}{2}$; for odd $s$ we conclude that 
$\check{c}_{i,j} = \frac{s-1}{2}$ and $\check{c}_{i,b_k-j+1} = \hat{c}_{i,j} 
= \hat{c}_{i,b_k-j+1} = \frac{s+1}{2}$. Since $\check{c}_{i,j} \le 
\overline{c}_{i,j} \le \hat{c}_{i,j}$, the fact follows.  \qed
\end{proof}
The second fact extends the first fact up to the last stage of our 
computation.
\begin{fact}\label{fct-22}
Let $s = height(\check{c})$. If $s$ is even then $\overline{c}_{i,j} =
\frac{s}{2}$ and $\overline{c}_{i,b_k-j+1}\in\{\frac{s}{2},\frac{s}{2}+1\}$
for each $i = 5k-11, \ldots, 6k-15$ and $j = 1, \ldots, k-2$. If $s$ is odd
then $\overline{c}_{i,j} \in\{\frac{s-1}{2},\frac{s+1}{2}\}$ and
$\overline{c}_{i,b_k-j+1} = \frac{s+1}{2}$ for each $i = 5k-11, \ldots, 6k-15$
and $j = 1, \ldots, k-2$.
\end{fact}
\begin{proof}
Consider first the sequence $\overline{c}_{5k-12}$ and observe that for $i =
5k-12$ the value of $\lceil\frac{i+1-(3k-6)}{2}\rceil$ is equal to $k-2$. It
follows from Fact \ref{fct-21} that for even $s$ all values from the left half
of $\overline{c}_{5k-12}$ are equal to $\frac{s}{2}$ and all values from the
right half of $\overline{c}_{5k-12}$ are in $\{\frac{s}{2}, \frac{s}{2}+1\}$.
For odd $s$ all values from the left half of $\overline{c}_{5k-12}$ are in
$\{\frac{s-1}{2},\frac{s+1}{2}\}$ and all values from the right half of
$\overline{c}_{5k-12}$ are equal to $\frac{s+1}{2}$. Since $Q^k_1$, $Q^k_2$ 
and $Q^k_3$ are built of functions $dec^k_*$, $mov^k_*$ and $cyc^k$ (cf. 
Definitions \ref{defFun} and \ref{defQ}) observe that each function $f_i$, $i 
= 5k-11, \ldots, 6k-15$ can only exchange values at positions from 
$args(mov^k_*)$ that are from non-constant half of arguments (in case of 
$dec^k_*$ and $cyc^k$ we can observe that for $a\le b\le a+1$ and any $h\ge 
0$ we have $\min(a,b+h) = a$, $\max(a-h,b) = b$, $\max(a,b-1) = a$ and 
$\min(a+1,b) = b$, that is, the functions are identity mappings in stages $5k-11, \ldots, 6k-15$). The $mov^k_*$ functions can only exchange unequal values 
at neighbour positions moving the smaller value to the left.  \qed
\end{proof}
The last fact states that unequal values $\overline{c}_{i,j}$ described in the
previous two facts are getting sorted during the computation. Observe that if
$s$ is odd (even, respectively) then we only have to trace the sorting process
in a left (right, respectively) region of indices
$[1,\min(k-2,\lceil\frac{i+1-(3k-6)}{2}\rceil)]$
($[\max(k-1,b_k-\lceil\frac{i+1-(3k-6)}{2}\rceil+1),b_k]$, respectively),
where $i = 3k-6, \ldots, 6k-15$ and the values to be sorted differs at most by
one. We trace the positions of the smaller values $s'=\frac{s-1}{2}$ in the
left region and the greater values $s'=\frac{s}{2}+1$ in the right region. We
will call $s'$ a moving element. For $t = 1, \ldots, k-2$ let us define $i_t =
3k+2t-8$ to be the stage, after which the length of the region extends from
$t-1$ to $t$ and a new element appears in it. Let $t'=t$ for odd $s$ and
$t'=b_k-t+1$, otherwise, be the position of this new element and
$a_t=c_{i_t,t'}$ be its value. Finally, let $n_t = |\{1\le l\le t | a_l =
s'\}|$ be the number of moving elements in the region after stage $i_t$.
\begin{fact}\label{fct-23}
Using the above definitions, for $t = 1, \ldots, k-2$, if $a_t = s'$ then for $i
= 0, \ldots, 6k - 15 - i_t$ we have $c_{i_t+i,\max(t-i,n_t)} = a_t$ if $s$ is 
odd and $c_{i_t+i,\min(t'+i,b_k-n_k+1)} = a_t$, otherwise.
\end{fact}
\begin{proof}
We prove the fact only for odd $s$, that is, for the left region. The proof
for the right region is symmetric. We would like to show that if $a_t = s'$
appears at position $t' = t$ after stage $i_t$ then it moves in each of the
following stages one position to the left up to its final position $n_t$. The
proof is by induction on $t$ and $i$. If $t=1$ and $a_1=s'$ appears at
position 1 after stage $i_1=3k-6$ then $n_1=1$ and $a_1$ is already at its
final position. It never moves, because values at second position are $\ge
s'$, by Facts \ref{fct-21} and \ref{fct-22}. If $t>1$ and $a_t = s'$ then the
basis $i=0$ is obviously true. In the inductive step $i>0$ we assume that 
$c_{i_t+i-1,\max(t'-i+1,n_t)} = a_t$ and that the fact is true for smaller 
values of $t$. If $\max(t-i+1,n_t) = n_t$ then also $\max(t-i,n_t) = n_t$ 
and, by the induction hypothesis, values at positions $1, \ldots, n_t-1$ 
are all equal $s'$. That means that $a_t$ is at its final position and we are 
done. Thus we left with the case: $n_t < t-i+1$, that is, with $n_t \le t-i$. 

Consider the sequences $\overline{c}_{i_t+i-1}$ and $\overline{c}_{i_t+i} =
f_{i_t+i}(\overline{c}_{i_t+i-1})$. We know that $\overline{c}_{i_t+i-1,t-i+1}
= s'$. To prove that $\overline{c}_{i_t+i,t-i} = s'$ we would like to show 
that $\overline{c}_{i_t+i-1,t-i} = s'+1$ and $mov^k_{t-i}\in f_{i_t+i}$. The 
later is a direct consequence of an observation that $mov^k_a\in f_b$ if and 
only if $(a+b)\equiv 1 (~mod~3)$. In our case $(t-i) + (i_t+i) = t+i_t = t + 
3k + 2t - 8 \equiv 1 (~mod~3)$. To prove the former, let us consider $a_u = 
s'$, $u\le t-1$. Then $i_u \le i_t-2$ and $n_u \le n_t-1$. By the induction 
hypothesis, $c_{i_u+j,\max(u-j,n_u)} = s'$. Setting $j = i_t - i_u + i-1$ we 
get $j\ge i+1$ and $\max(u-j,n_u) \le \max(t-1-(i+1),n_t-1) < \max(t-i,n_t) = 
t-i$. Moreover, $i_u+j = i_t+i-1$. That means that in the sequence  
$\overline{c}_{i_t+i-1}$ none of $n_t$ elements $s'$ is at position $t-i$ 
and, consequently, $\overline{c}_{i_t+i-1,t-i} = s'+1$. Since $mov^k_{t-i}$ 
switches $s'$ with $s'+1$, this completes the proof of Fact \ref{fct-23}.  \qed
\end{proof}
Now we are ready to prove that $\overline{c}_{6k-15}$ is a flat sequence. By
Fact \ref{fct-22}, if $s$ is odd then $\overline{c}_{6k-15} \in
\{\frac{s-1}{2},\frac{s+1}{2}\}^{k-2}(\frac{s+1}{2})^{k-2}$, otherwise,
$\overline{c}_{6k-15} \in
(\frac{s}{2})^{k-2}\{\frac{s}{2},\frac{s}{2}+1\}^{k-2}$. The number of minority
elements in $\overline{c}_{6k-15}$ has been denote by $n_{k-2}$. If $s$ is odd
and $a_t$, $t = 1, \ldots, k-2$, is a minority element $\frac{s-1}{2}$, then, by
Fact \ref{fct-23}, $c_{6k-15,n_t} = \frac{s-1}{2}$. If $s$ is even and $a_t$, $t
= 1, \ldots, k-2$, is a minority element $\frac{s}{2}+1$, then, by Fact
\ref{fct-23}, $c_{6k-15,b_k-n_t+1} = \frac{s}{2}+2$. In both cases this proves
that  $\overline{c}_{6k-15}$ is flat, which completes the proof of Theorem
\ref{thm-19}. \qed
\end{proof}

\subsection{Proof of Theorem \ref{3merger}}  

Theorem \ref{3merger} follows directly from Theorem \ref{thm-19} and Lemma
\ref{l3}. Let $k\ge 3$ and $\overline{c}$ be any 2-flat sequence of integers
from $[0,2^{k-1}-1]$. By Theorem \ref{thm-19} the result of application
$(Q^k_3\circ Q^k_2\circ Q^k_1)^{2k-5}$ to $(\overline{c})$ is a flat sequence.
Then, by Lemma \ref {l3}, the network $M_k$ is a $2k-5$-pass merger of two 
sorted sequences given in odd and even registers, respectively.

\section{Conclusions}

For each $k \ge 3$ we have shown a construction of a 3-periodic merging
comparator network of $N_k = 2^k(k-2)$ registers and proved that it merge any
two sorted sequences (given in odd and even registers, respectively) in time
$6k-15 = 3(k-5)$. A natural question remains whether it is the optimal merging
time for 3-periodic comparator networks.

\end{document}